\documentclass[10pt,pra,aps,twocolumn,amsmath,amssymb,footinbib,showpacs,longbibliography,superscriptaddress,linenumbers]{revtex4-1}

\usepackage{times}
\usepackage[english]{babel}
\usepackage{siunitx}
\usepackage[version=4]{mhchem}
\usepackage{latexsym}
\usepackage{graphics}
\usepackage{subfigure}
\usepackage{epsfig}
\usepackage{color}
\usepackage{soul}
\usepackage{hyperref}
\usepackage{braket}
\usepackage[T1]{fontenc}
\usepackage[utf8]{inputenc}
\usepackage{amstext}
\usepackage{amssymb}
\usepackage{graphicx}
\usepackage{esint}
\usepackage{braket}
\usepackage{babel}
\usepackage{amsmath}
\DeclareSIUnit{\gauss}{G}
\usepackage{mathrsfs}
\usepackage[normalem]{ulem}
\usepackage{amsthm}  % 定理环境支持

\newtheorem{lemma}{Lemma}[section] 

\hypersetup{
colorlinks=true,
citecolor=blue,
linkcolor=red,
urlcolor=black
}


\begin{document}
\nolinenumbers
\title{
Observation of a tripartite quantum phase for coexisting extended, localized, and critical states
}
\author{Zhongshu Hu}
% \thanks{These authors contributed equally to this work.}
\affiliation{International Center for Quantum Materials, School of Physics, Peking University, Beijing 100871, China}
\affiliation{Hefei National Laboratory, Hefei 230088, China}

\author{Yajing Guo}
\affiliation{International Center for Quantum Materials, School of Physics, Peking University, Beijing 100871, China}
\affiliation{Hefei National Laboratory, Hefei 230088, China}

\author{Yu-Dong Wei}
\affiliation{International Center for Quantum Materials, School of Physics, Peking University, Beijing 100871, China}
\affiliation{Hefei National Laboratory, Hefei 230088, China}

\author{Bing-Chen Yao}
\affiliation{International Center for Quantum Materials, School of Physics, Peking University, Beijing 100871, China}
\affiliation{Hefei National Laboratory, Hefei 230088, China}

\author{Zhentian Qian}
\affiliation{Cavendish Laboratory, University of Cambridge, JJ Thomson Avenue, Cambridge CB3 0HE, United Kingdom}

\author{Xin-Chi Zhou}
\affiliation{Max Planck Institute for the Physics of Complex Systems, N\"othnitzer Str. 38, 01187 Dresden, Germany}

\author{Bao-Zong Wang}
\affiliation{International Center for Quantum Materials, School of Physics, Peking University, Beijing 100871, China}
\affiliation{Hefei National Laboratory, Hefei 230088, China}

\author{Jianing Yang}
\affiliation{International Center for Quantum Materials, School of Physics, Peking University, Beijing 100871, China}
\affiliation{Hefei National Laboratory, Hefei 230088, China}

\author{Xuzong Chen}
\affiliation{School of Electronics, Peking University, Beijing 100871, China}
\affiliation{Hefei National Laboratory, Hefei 230088, China}

\author{Shengjie Jin}
\email{jinshengjie@pku.edu.cn}
\affiliation{International Center for Quantum Materials, School of Physics, Peking University, Beijing 100871, China}
\affiliation{Hefei National Laboratory, Hefei 230088, China}

\author{Xiong-Jun Liu}
\email{xiongjunliu@pku.edu.cn}
\affiliation{International Center for Quantum Materials, School of Physics, Peking University, Beijing 100871, China}
%\affiliation{International Quantum Academy, Shenzhen 518048, China}
\affiliation{Hefei National Laboratory, Hefei 230088, China}

%\date{\today}

\begin{abstract}
The disordered quantum world hosts three fundamental types of states: extended, localized, and critical, of which the critical states are confined to fine-tuned critical points or mobility edges in randomly disordered systems. The tripartite phase, with all three types of states coexisting over finite spectral windows, represents a hallmark distinction between quasiperiodic and truly random systems in the localization physics. %are separated by anomalous mobility edges.
Here, we report the realization of this exotic phase in a quasi-periodically driven orbital optical lattice with ultracold atoms. The optical lattice with a quasiperiodic Floquet modulation coupling $s$ and $p$ orbitals is realized in experiment and shown to host the tripartite phase from exact theory. We develop a two-stage protocol to precisely prepare and detect the three types of quantum states. The characteristic exponents of these states are determined from expansion dynamics, showing their distinct universal transport properties. Our study marks a significant advancement in exploring unconventional critical phenomena and localization physics with ultracold atoms.
\end{abstract}

\maketitle

% \linenumbers

\section{Introduction}

The quantum world offers a rich tapestry of states beyond the familiar dichotomy of metals and insulators. At the heart of this complexity lies Anderson localization—a phenomenon where disorder can halt the extension of quantum wave-functions of particles, transforming a conductor into an insulator~\cite{Anderson1958,Abrahams1979PRL,PatrickALee1985RMP,Kramer1993RPP,Evers2008RMP}. Anderson localization stands as a fundamental concept in condensed matter physics, inspiring active studies for over one-half century~\cite{Abrahams2010}, is crucial for various prominent phenomena such as quantum Hall effect~\cite{Laughlin1981,Halperin1982}, and is deeply connected to the topological classification theory~\cite{PhysRevB.55.1142,Ryu_2010}. %with its significance rivaling superconductivity in manifesting the quantum coherence at macroscopic scale, and has been actively studied for one-half century~\cite{Abrahams2010}.
For the past two decades, ultracold atoms have emerged as a versatile quantum platform %for exploring these phenomena,
enabling controlled studies of Anderson transitions~\cite{Roati2008Nature,Billy2008Nature}, mobility edges~\cite{Semeghini2015NP,Luschen2018PRL}, and more recently, many-body localization where the interplay between interactions and disorder prevents thermalization—a fundamentally new phase of matter~\cite{Schreiber2015,Choi2016Science,Abanin2019RMP}.

However, nature is not limited to the binary outcome of extended versus localized states. Between the two lies a third, more enigmatic class: critical states, which in randomly disordered systems emerges at fine-tuned critical points or mobility edges~\cite{Evers2008RMP}. These states are neither confined to a region nor spread uniformly across the entire system~\cite{LiuNatPhys2024,huang2026experimental}. Instead, they exhibit a unique, scale-invariant geometry known as multifractality~\cite{PhysRevLett.97.046803,PhysRevLett.112.234101}, residing on a complex, self-similar set that parallels the fractal patterns observed across diverse natural phenomena, from snowflake formation~\cite{Libbrecht2005RPP} to galaxy distributions~\cite{coleman1992fractal}. This multifractal character, a hallmark of systems poised at a critical point, makes critical states exquisitely sensitive to their environment, a feature recently recognized as a powerful resource for quantum sensing and %criticality-enhanced
metrology~\cite{Sahoo2024PRA,Candia2023npj,Ilias2022PRXQ}. Such scale-invariant behavior also plays nontrivial roles in ground state properties associated with exotic symmetry breaking~\cite{Feigelman2007,Burmistrov2012,Zhao2019Disorder,Sacepe2020,Goncalves2024} and emergence of the profound non-ergodic many-body critical phase~\cite{Wang2020PRL,Wang2021PRL} that defies eigenstate thermalization hypothesis~\cite{DAlessio2016,Rigol2008,Srednicki1994}.

The coexistence of extended, localized, and critical states within a single quantum phase defines an exotic regime that is fundamentally inaccessible in randomly disordered systems. A striking example is the tripartite phase in quasiperiodic lattices~\cite{Wang2022PRB,Zhou2026,Goncalves2023Critical}, where all three types of states coexist over finite spectral windows, separated by anomalous mobility edges~\cite{Goncalves2023Critical,TongLiu2022Anomalous,Zhou2023PRL}. Remarkably, this phase represents a hallmark distinction between the two canonical classes of disordered systems: quasiperiodic and random, and is beyond the characterization of standard $\sigma$-model~\cite{Evers2008RMP}. It also provides a fertile ground for exploring a wealth of new physics, including unconventional transport dynamics~\cite{Landi2022RMP,Wang2020PRL,Wang2022PRB}, new critical transitions~\cite{Rispoli2019,Goblot2020NP,Li2023npj,shimasaki2024anomalous}, and novel thermal–nonthermal crossovers in interacting systems~\cite{PhysRevB.82.174411,PhysRevB.97.174206,Abanin2019RMP}. Recent theoretical advances have established a comprehensive phase diagram of localization physics in quasiperiodic systems, revealing the seven fundamental localization phases and predicting the conditions under which the elusive tripartite phase may emerge~\cite{Zhou2026}. Nevertheless, thus far this exotic phase remains hidden from experimental view. The challenge lies in not only its realization but also the precise preparation and detection of the high-energy eigenstates that define the phase, a task that is typically difficult for many-particle quantum systems. %such as solids and ultracold atoms. %in the solid-state systems and ultracold atom platforms.

In this article, we report the first experimental observation of the novel tripartite quantum phase %for coexisting extended, localized, and critical states
in a quasiperiodic orbital optical lattice. %with incommensurate Floquet modulation.
The realized model is not analytically solvable, whereas we show rigorously with renormalization group method that it hosts critical states along with extended and localized states. We develop a novel two-stage protocol to selectively prepare ultracold bosons in different types of states separated by anomalous mobility edges, with the nature of these states revealed by measuring their characteristic lengths in both real and dual momentum spaces. The direct measurement of the universal transport dynamics further identify the three classes of states and show their fundamental features.
This work successfully achieves a tripartite quantum phase for ultracold atoms and provides a benchmark to explore exotic quasiperiodic quantum matter.
% In this work, we report the first observation of a tripartite quantum phase in which extended, localized, and critical states coexist, realized in an experiment with ultracold atoms. This is achieved by implementing an orbital quasiperodic lattice, which couples  the two lowest Bloch bands(s and p orbitals) via a secondary lattice modulation. To access the different regimes of the spectrum, we develop a two-stage preparation protocol that enables selective loading atoms into three types of states. The nature of the resulting states is revealed by combining the real-space localization length with the momentum-space density–density correlation length, which together allow us to construct a comprehensive final-state phase diagram. Moreover, expansion dynamics directly distinguish the three phases, yielding characteristic dynamical indices in agreement with theoretical expectations and providing a benchmark for future studies of quasiperiodic quantum matter.

\begin{figure*}
    \centering
    \includegraphics[width=1\textwidth]{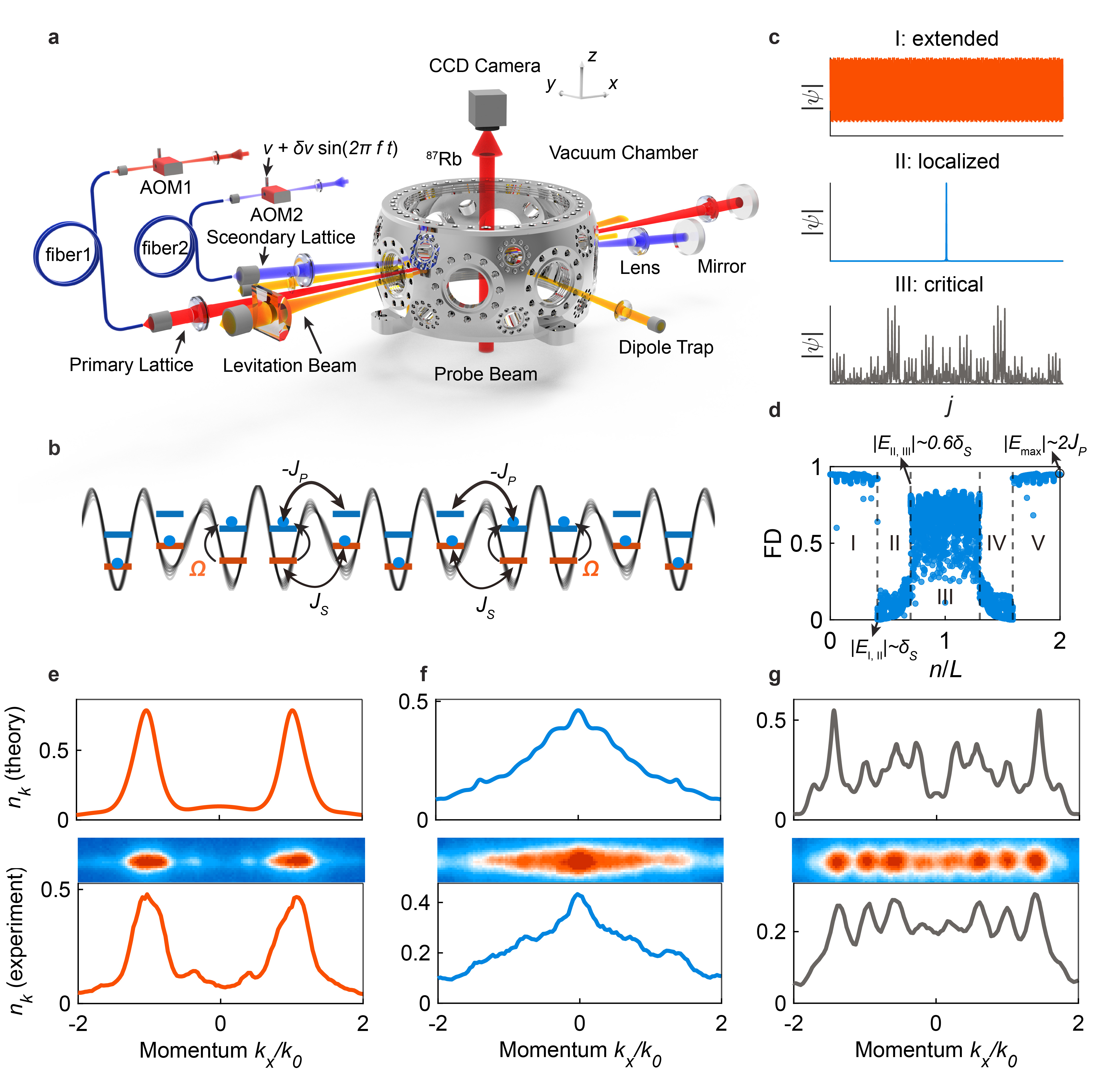}
    \caption{\textbf{Schematic of the incommensurately driven $s$-$p$ orbital lattice and characteristic states.}
    \textbf{a}, Experimental configuration. A Bose-Einstein condensate is loaded into a one-dimensional quasiperiodic optical lattice (along the $x$ direction) generated by superimposing a primary ($\lambda_0 = 1064$ nm) and a secondary ($\lambda_1 = 760$ nm) lattice. Absorption imaging is performed along the $z$ direction. The relative displacement of the lattices is controlled by modulating the driving frequency of the acoustic-optic modulator (AOM) for the secondary lattice. A levitation beam, shaped as a sheet along the $x$ direction, provides a strong vertical gradient force to counteract gravity along the $z$ direction, while maintaining low trap frequencies in the $xy$ plane.
    \textbf{b}, Lattice structure schematic. A near-resonant Floquet modulation of the secondary lattice hybridizes the $s$ and $p$ orbitals, inducing an incommensurate inter-orbital coupling $\Omega$. $J_s$ and $J_p$ denote the respective intra-orbital tunneling amplitudes.
    \textbf{c}, Illustrative real-space distributions of the three types of localization states, obtained from theory.
    \textbf{d}, Fractal dimension (FD) of the effective Hamiltonian eigenstates versus the eigenenergy $E/J_P$ for a system size $L=10946$. Sharp transitions in FD demarcate five regions, where regions I, II, and III represent the extended, localized, and critical zones in the tripartite phase. The MEs are characterized by the conventional ME at $|E_{I,II}| \sim \delta_s$, while the anomalous ME separating localized and critical states emerges at $|E_{II,III}| \sim 0.6 \delta_s$.
    \textbf{e}–\textbf{g}, Representative momentum distributions. Top: Theoretical $n_k$ averaged over a finite energy window. Middle: Experimental TOF images ($t_{\mathrm{TOF}} = 45$ ms). Bottom: Integrated 1D momentum profiles, showing the distinct fingerprints of extended (sharp peaks), localized (broad), and critical (fragmented multi-peak) states.}
    \label{fig:1}
\end{figure*}

\section{$S$-$P$ Orbital Quasiperiodic Lattice}

The model we propose and realize is termed $s$-$p$ orbital quasiperiodic lattice---a one dimensional system where $s$ and $p$ orbitals are coupled through a quasiperiodic Floquet modulation (Methods)~\cite{Supplementary}. %quasiperiodic orbital lattice, which applies a quasiperiodic Floquet modulation to couple the $s$ and $p$ orbitals, termed the $s$-$p$ orbital quasiperiodic lattice. %We start with a one dimensional lattice model that hosts a tripartite quantum phase in which extended, localized, and critical states coexist, which we term the $s$-$p$ orbital quasiperiodic lattice, where the s and p orbitals are the two lowest orbitals that are coupled through an incommensurate lattice modulation in an optical lattice.
The model system is realized by imposing two nearly collinear optical lattices of different wavelengths, with one being deep primary and another being weak secondary, as shown in Fig.~\ref{fig:1}a. The $s$ and $p$ orbitals correspond to the lowest two energy levels of the primary lattice at each site~\cite{Li2016RPP}. Our key ingredient is that the secondary lattice is spatially incommensurate and shaken with respect to the primary one at a frequency matching the $s$-$p$ orbital energy gap. The incommensurate shaking marks a major difference of the present scheme from the previously widely studied spatially periodic shaken lattices in the literature~\cite{Jotzu2014,flaschner2016experimental,schweizer2019floquet}. %Resonant shaking, by directly coupling the two lowest orbitals, provides an especially simple and controllable scheme, which has already been demonstrated in weakly interacting systems~\cite{Parker2013,Songbo2022,Sun_2023}.
In the deep lattice regime, we write down the tight-binding Hamiltonian by considering only the onsite and nearest neighboring couplings, which takes the form in the rotating frame as $ \hat{H} = \hat{H}_0 + \hat{H}_{\mathrm{qp}}$. Here $\hat{H}_0$ describes the bare $s$-$p$ orbital lattice and $\hat{H}_{\mathrm{qp}}$ captures the quasiperiodic modulation \cite{Supplementary}, given by
\begin{eqnarray}
\begin{split}
\hat{H}_0 &= \sum_j \Big( -J_s \hat{a}_j^\dagger \hat{a}_{j+1} + J_p \hat{b}_j^\dagger \hat{b}_{j+1} +\text{H.c.} \Big) +\Delta \hat{b}_j^\dagger \hat{b}_j, \\
\hat{H}_{\mathrm{qp}} &= \sum_j  \mathrm{i}\Omega_j (\hat{a}_j^\dagger \hat{b}_j+\hat{b}_j^\dagger \hat{a}_j) + \big( \delta_{s,j}\hat{a}_j^\dagger \hat{a}_j + \delta_{p,j}\hat{b}_j^\dagger \hat{b}_j \big),
\end{split}
\label{eq:H}
\end{eqnarray}
where $\hat a_j$ ($\hat b_j$) is the bosonic annihilation operator of $s$ ($p$) orbital, and $J_p \gg J_s$ are the nearest-neighbor hopping coefficients. The quasiperiodic modulation leads to both inter-orbital onsite quasiperiodic coupling $\Omega_j=\Omega \sin(2\pi \beta j)$ and intra-orbital onsite quasiperiodic potentials $\delta_{\nu,j}=\delta_{\nu} \cos(2\pi \beta j)$, with $\nu=s,p$ for the $s$ and $p$ orbitals, respectively (Fig.~\ref{fig:1}b). The effective detuning $\Delta=\Delta_E-hf$ is determined by the average gap $\Delta_E$ between the two orbitals and the shaking frequency $f$, with $h$ being the Planck constant.
%Additionally, we have performed simulations with Gross-Pitaevskii equations showing that higher bands are off-resonant and unimportant for the used experimental sweeps.

The Hamiltonian in Eq.~(\ref{eq:H}) is not exactly solvable due to the presence of quasiperiodic couplings, but the realization of the tripartite phase can be rigorously proved. Before giving a quantitative proof, we elaborate the intuition behind the emergence of this phase. For convenience we take the $s$-orbital to be frozen with $J_s\rightarrow0$, and other parameters $J_p\geq\delta_s\sim \delta_p>\Omega$. For eigenstates residing in the band edges, with eigen-energies $|E|>\delta_s$ beyond $s$ orbital onsite energy, they are mainly contributed by the $p$ orbital, and are extended since the hopping coupling $J_p$ dominates over quasiperiodic terms. On the other hand, for eigen-energy within $s$ orbital onsite energy amplitude $|E|\leq\delta_s$, the inter-orbital coupling in Eq.~(\ref{eq:H}) can enable full transition from $p$ to $s$ orbitals at the $\tilde j$-sites with $\delta_{s,\tilde j}=E$. These sites are incommensurately distributed over the system. The full transition to $s$-orbital impedes tunneling of the $p$ orbitals across those sites, manifesting incommensurately distributed zeros (IDZs) in hopping, a central mechanism leading to the critical states and captured by the profound Avila's global theory~\cite{Zhou2023PRL}. With the IDZs, there are two scenarios. For $|E|\lesssim\delta_s$, the states reside in the vicinity of maximum/minimum of the onsite potential $\delta_{s,j}$, with relatively small inter-orbital coupling $\Omega_j$. These states are strongly affected by the qusiperiodic potential and are localized. For $E\sim0$, the states are mainly located at the nodes of the qusiperiodic potential $\delta_{s,j}$, with relatively large $\Omega_j$. These states are weakly affected by quasiperiodic potential, hence being delocalized and critical. More rigorously, we show with renormalization group method that the critical state at zero energy $E\rightarrow0$ emerges when ($\Omega\neq0$)
\begin{equation}
\Omega^2 + \delta_s\delta_p \le 2\delta_s J_p,
\end{equation}
marking the regime in which critical states persist~\cite{Supplementary}.

The three types of states can be numerically identified by the fractal dimension (FD). The FD approaches $1$ ($0$) for extended (localized) states, while $0<\mathrm{FD}<1$ characterizes critical states~\cite{Zhou2023PRL}. Taking the parameters $J_p=\delta_s=\delta_p=2\Omega=1$, $J_s=0.1$, and $\Delta=0$, we show in Fig.~\ref{fig:1}d the FD of each eigenstate versus the eigenenergy $E/J_p$. Five distinct regions with MEs characterized by sharp transitions in FD separating extended, localized, and critical zones, corresponding to regions I, II, and III in the figure, consistent with the above analysis of the tripartite phase.

These states can be detected based on their distinct localization properties~\cite{Wang2022PRB}. In real space, localized states are confined near a single lattice site, extended states spread across the entire system, and critical states display self-similar fractal structures [Fig.~\ref{fig:1}c]. In momentum space, localized states give rise to broad distributions [Fig.~\ref{fig:1}f, upper row], whereas extended states yield sharp peaks [Fig.~\ref{fig:1}e, upper]. Critical states show intermediate features: their eigenfunctions are multifractal and delocalized, leading to momentum distributions that are clearly distinct from both [Fig.~\ref{fig:1}g]. In experiments, finite temperature further broadens the measured distributions, which can be modeled as a thermal superposition of eigenstates within a finite energy window.

\section{Experimental Setup}
The experimental realization of the $s$-$p$ orbital quasiperiodic lattice is based on a Floquet engineering of ultracold $^{87}$Rb atoms in a bichromatic optical potential.
First, a Bose-Einstein condensate of $1\times 10^5$ \ce{^{87}Rb} atoms is produced by evaporative cooling in an optical dipole trap, without discernible thermal component. The atoms are polarized in the hyperfine state $\left|F=2, m_F=2\right\rangle$.
% The trapping frequencies are $\nu_x=\nu_y=\SI{20}{Hz}$ and $\nu_z=\SI{28}{Hz}$, and gravitational sag is compensated via magnetic levitation.
The condensate is then adiabatically loaded into a bichromatic optical lattice formed by superimposing a primary static lattice with wavelength $\lambda_0=\SI{1064}{nm}$ and a secondary shaken lattice with wavelength $\lambda_1=\SI{760}{nm}$. As illustrated in Fig.~\ref{fig:1}a, the primary lattice is aligned along $x$ direction, and the secondary lattice lies in $xz$ plane at an angle of $\theta=7^\circ$ relative to $x$ axis. The lattice depths are set to $10 E_{r,0}$ (primary lattice) and $1.6 E_{r,1}$ (secondary lattice), ensuring the tight-binding regime. Here the recoil energies $E_{r,0}=\hbar^2 k_0^2 / (2m)$ and $E_{r,1}=\hbar^2 k_1^2 / (2m)$, with $k_0=2\pi/\lambda_0$ and $k_1=2\pi/\lambda_1$. The spatial incommensurability is introduced by the secondary lattice with a wavelength ratio $\beta = \lambda_0/(\lambda_1 \cos \theta)\approx 1.411$, a value sufficiently irrational to emulate quasiperiodicity. The energy gap between the $s$ and $p$ orbitals is $\Delta_E=\SI{10.2}{kHz}$, with tunneling amplitudes $J_s=\SI{39}{Hz}$ and $J_p=\SI{505}{Hz}$.

Both lattices are ramped up simultaneously within \SI{80}{ms} following a standard exponential loading sequence, preparing the atoms in a localized state in the $s$ orbital. After holding for \SI{13}{ms}, periodic modulation of the secondary lattice position is initiated. This is achieved by sinusoidally modulating the driving frequency of the secondary lattice laser using an acousto-optic modulator in a double-pass configuration. Owing to the retro-reflected geometry, as shown in Fig.~\ref{fig:1}a, this frequency modulation translates into a spatial oscillation of the lattice at the atomic position. The $s$-$p$ orbital coupling is engineered by sinusoidally modulating the position of the secondary lattice at a frequency $f \approx \Delta_E/h$. This quasiperiodic driving translates the spatial incommensurability into a site-dependent inter-orbital coupling $\Omega_j$.The resulting coupling strength $\Omega$ is approximately proportional to the shaking displacement amplitude $A$, which is determined by the modulation depth and the distance $l_0\approx \SI{0.9}{m}$ between the atoms and the retro-reflector.

A central challenge in observing the tripartite phase is the selective precise preparation of cold atoms into target high energy eigenstates, which is inaccessible via standard adiabatic loading. We solve this challenge by developing a novel two-stage protocol that utilizes both frequency and amplitude ramps to prepare the system in target energy windows, with the observed distinct momentum distributions signifying the coexisting three types of states [lower panels of Fig.~\ref{fig:1}e-g]. This powerful protocol is applicable to exploring broad localization physics in Floquet quantum systems. %In the following, we provide a detailed exposition of the dynamic control sequences employed to achieve the preparation of these three distinct quantum states.

\begin{figure*}
    \centering
    \includegraphics[width=0.9\textwidth]{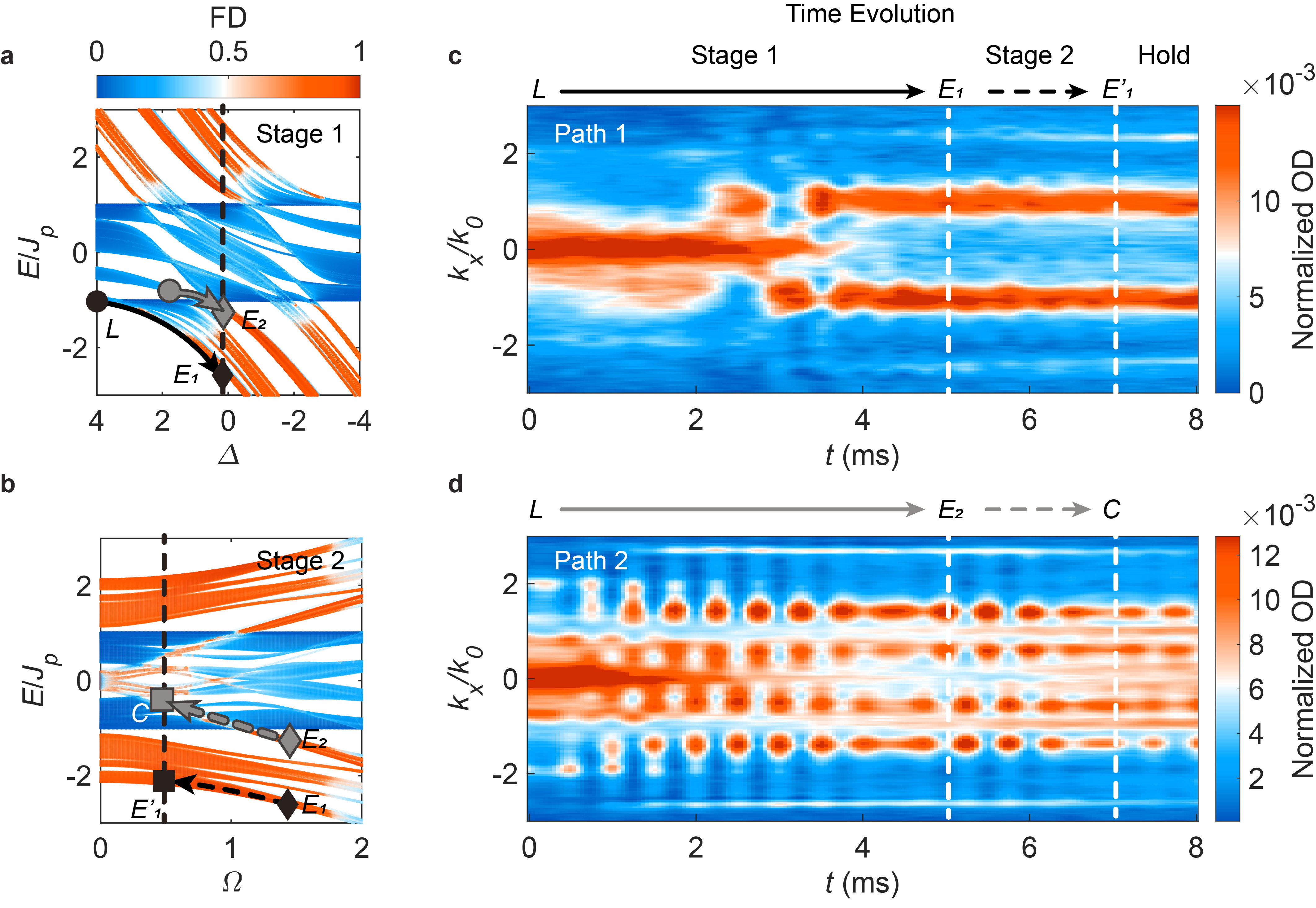}
    \caption{\textbf{Two-stage protocol for selective state preparation.}
    \textbf{a},\textbf{b}, Fractal dimensions (FDs) of all eigenstates as functions of the detuning $\Delta$ (\textbf{a}) and coupling strength $\Omega$ (\textbf{b}), serving as schematic representations of the evolution of quantum states. The color scale denotes the FD, where Path 1 (black) and Path 2 (gray) represent the distinct evolutionary trajectories. Stage 1 is executed in (\textbf{a}) via an adiabatic frequency sweep (solid lines) to bridge the $s$-$p$ orbital gap, while Stage 2 is performed in (\textbf{b}) through an amplitude ramp (dashed lines). In these processes, the initial localized state $L$ evolves into intermediate extended states $E_1$ and $E_2$, respectively. Subsequently, $E_1$ adiabatically maps onto the ground extended state $E_1'$, whereas $E_2$ serves as a stepping stone, diabatically evolving into the critical regime $C$ as $\Omega$ decreases. The vertical dashed lines indicate the parameters of the final Hamiltonian.
    \textbf{c},\textbf{d}, Time-resolved momentum evolution along Path 1 (\textbf{c}) and Path 2 (\textbf{d}). In Path 1, $L$ adiabatically transforms into the ground extended state $E_1$, evidenced by narrow peaks at $\pm k_0$. In Path 2, the system evolves from $L$ to the excited extended state $E_2$ (four-peak structure) and finally the critical state $C$, characterized by a fragmented distribution across the first two Brillouin zones. TOF images are recorded every \SI{0.25}{ms}, averaged over six repetitions, and smoothed between consecutive points. }
    \label{fig:2}
\end{figure*}

\section{Two-stage preparation protocol}
Our experimental objective is twofold: to engineer the Hamiltonian in the parameter regime that realizes the tripartite phase for coexisting extended, localized, and critical states, and to selectively populate these distinct eigenstates with cold bosons. We establish a two-stage protocol to achieve this important goal by engineering the two intrinsic controlling knobs of the Floquet system, the modulation frequency $f$ and amplitude $A$.
The first stage utilizes an adiabatic frequency ramp to %navigate the Floquet spectrum and
lock the system into a specific energy sector. The second stage is a rapid amplitude ramp which projects the system into the final parameter regime hosting tripartite phase ($A_f = \SI{21}{nm}$, $f_f = \SI{10.2}{kHz}$, giving $\Omega\sim 0.5$ and $\Delta\sim0$ in the tight-binding model). This protocol enables the realization of the tripartite phase with high-fidelity state control.

For the first stage, we employ an adiabatic frequency ramp ($f_i\rightarrow f_f$) to prepare the system into specific energy sectors. We start with the initial frequency $f_i$ and shaking amplitude $A_i$, and then the frequency is slowly ramped to $s$-$p$ resonant regime with $f_f = \SI{10.2}{kHz}$ ($\Delta=0$) over \SI{4}{ms} while keeping $A_i$ (Methods). We identify three distinct preparation pathways based on the coupling strength $\Omega$ and initial frequency $f_i$, which respectively drive the system into different target states in the fully controllable way. %This two-stage protocol ensures that the system is precisely prepared onto different target states, determined by the selection of initial parameters.
First, for small $A_i$ (weak coupling $|\Omega/J_p|\ll1$), the $s$-$p$ hybridization is negligible. Consequently, the system remains in the localized $s$-orbital state throughout the ramp, independent of $f_i$.
Second, for large $A_i$ (strong coupling, e.g., $\Omega/J_p \approx 1.2$), the evolution bifurcates based on initial frequency $f_i$, as illustrated in Fig.~\ref{fig:2}a. With a low $f_i $ (Path 1), the $s$-band intersects the $p$-band bottom during the ramp. This induces a Landau-Zener transition, transferring atoms into the low-energy extended state $E_1$. Conversely, with a high $f_i$ (Path 2), the $s$-band initially overlaps and hybridizes with the $p$-band in the middle. No band-edge crossing occurs, and the system evolves adiabatically into the high-energy extended state $E_2$.

% The second stage implements a diabatic state projection to reach the final Hamiltonian and obtain all three types of states. We rapidly reduce the shaking amplitude from different $A_i$ to the same $A_f = \SI{21}{nm}$ within \SI{2}{ms} while maintaining $f_f$.
% Figure~\ref{fig:2}(\textbf{B}) illustrates the eigenenergies and FDs of Hamiltonian~\ref{eq:H} as a function of $\Omega$, at the resonance $f_f = \SI{10.2}{kHz}$ ($\Delta = 0$). For small $A_i$ case, the system stays in a localized state which doesn't shown in the figure. For two large $A_i$ cases, the paths are shown as follows.
% For Path 1 (black arrow in Fig. \ref{fig:2}(\textbf{B})), the state $E_1$ remains stable and maps onto the ground extended state $E_1'$ of the final Hamiltonian. Crucially, for Path 2 (gray arrow in Fig. \ref{fig:2}(\textbf{B})), the rapid amplitude ramp acts as a diabatic drive, steering the stable extended state $E_2$ into the energy window hosting the critical states $C$ [Fig. \ref{fig:2}(\textbf{B})]. This strategy utilizes $E_2$ as a "stepping stone", allowing us to precisely target the tripartite phase's most exotic component.

% The second stage employs a diabatic projection to reach the final Hamiltonian while targeting specific eigenstates. We reduce the shaking amplitude from different initial values $A_i$ to a common final value $A_f = \SI{21}{nm}$ within \SI{2}{ms}, keeping the frequency at $f_f$.
For the second stage, we implement a diabatic state projection to reach the final Hamiltonian and obtain the targeted type of states by rapidly reducing the shaking amplitude as $A_i\rightarrow A_f = \SI{21}{nm}$ within \SI{2}{ms} while maintaining $f_f$ (Methods). As shown in the spectrum versus $\Omega$ [Fig. \ref{fig:2}b], the state $E_1$ (Path 1) remains stable and adiabatically evolves into the ground extended state $E_1'$ of the final Hamiltonian. Crucially, in Path 2, the amplitude ramp acts as a diabatic drive that steers the intermediate extended state $E_2$—acting as a "stepping stone"—into the energy window hosting the critical states $C$. For the small $A_i$ case, the system remains in the localized regime throughout the amplitude ramping.

We show the high feasibility of the two-stage protocol with the evolution of momentum distributions measured from time-of-flight (TOF) absorption imaging [Fig. \ref{fig:2}c,d]. In Path 1 ($f_i = \SI{8}{kHz}$), the broad distribution of the initial localized state $L$ evolves into two sharp peaks at $\pm k_0$ [Fig. \ref{fig:2}c], giving the ground extended state $E_1'$ [Fig. \ref{fig:1}e]. In Path 2 ($f_i = \SI{10}{kHz}$), the system first evolves into a characteristic four-peak structure at $\pm(2k_0 - k_1)$ and $\pm k_1$, which is the signature of the hybridized state $E_2$. During the second stage, this structure transforms into a fragmented, multi-peak distribution across the first two Brillouin zones [Fig. \ref{fig:2}d], being the features of a critical state $C$ [Fig. \ref{fig:1}g]. The oscillations observed between $\pm k$ arise from micromotion inherent to the Floquet modulation \cite{Eckardt2017Colloquium, Bukov2015high-frequency, Songbo2022,Supplementary}.
With the fully controlled initial parameters $(A_i, f_i)$, we can navigate the system to be precisely prepared in extended ($E_1'$), critical ($C$), or fully localized state. In particular, the metastable $E_2$ state provides a reliable pathway to access the critical states, whichis difficult to populate via standard adiabatic loading.

\begin{figure*}
    \centering
    \includegraphics[width=1\textwidth]{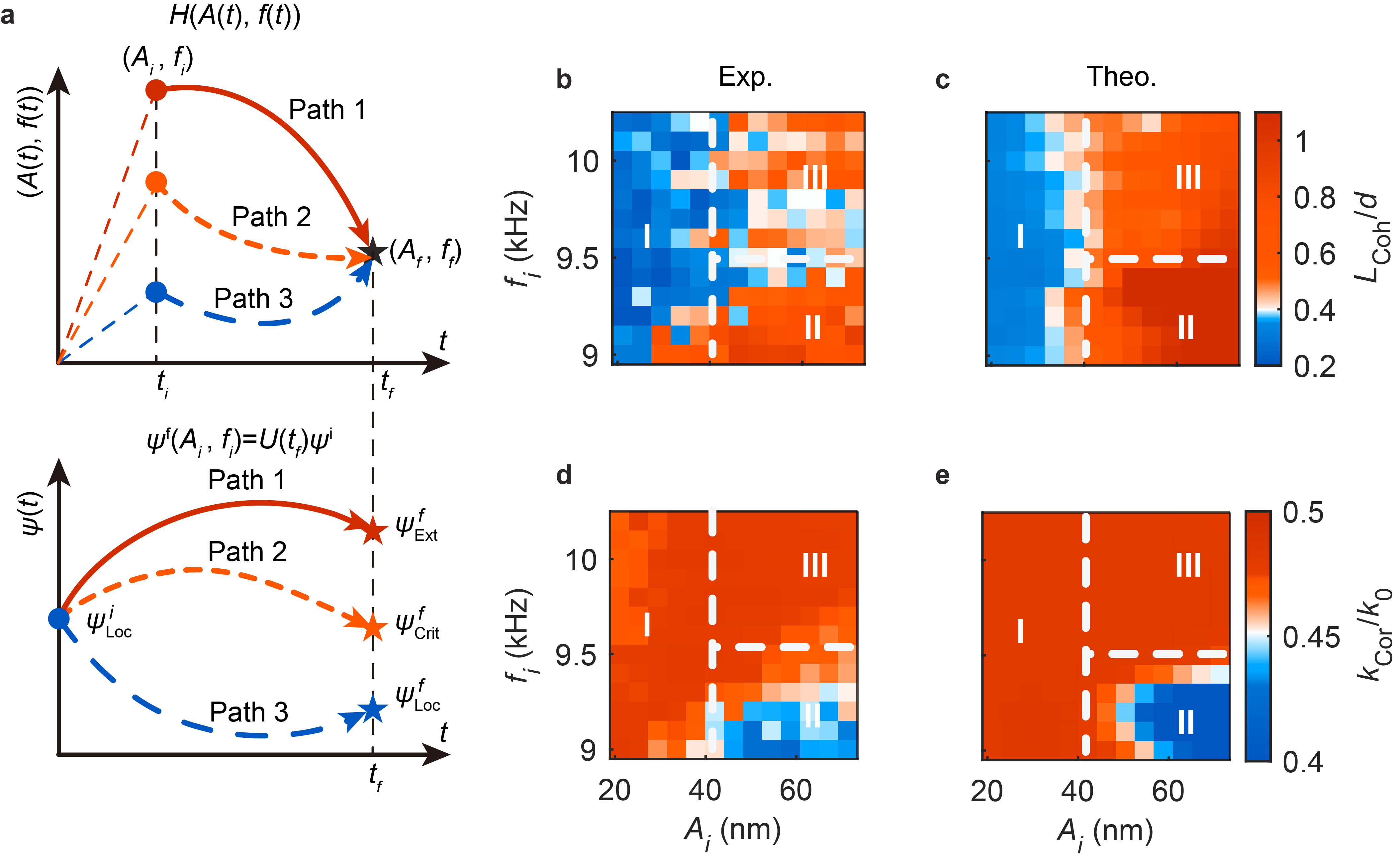}
    \caption{
\textbf{Quantitative identification of the tripartite phase via characteristic lengths.}
\textbf{a}, Schematic illustration of the pathway dependence of the final states. The upper panel shows the trajectory in the $(A, f)$ parameter space, and the lower panel depicts the corresponding state evolution $\psi(t)$. All paths terminate at the same final Hamiltonian $(A_f,f_f)$ (upper), with different choices of $(A_i, f_i)$ rendering different types of target high energy states (lower).
\textbf{b},\textbf{c}, Spatial coherence length $L_{\mathrm{Coh}}$ from experiment (\textbf{b}) and simulation (\textbf{c}) as a function of the initial shaking frequency $f_i$ and amplitude $A_i$. $L_{\mathrm{Coh}}$ quantifies spatial extent; its suppression in region I identifies the localized regime.
\textbf{d},\textbf{e}, Momentum-space density-density correlation length $k_{\mathrm{Cor}}$ from experiment (\textbf{d}) and theory (\textbf{e}) under the same driving conditions. $k_{\mathrm{Cor}}$ quantifies the spread of momentum distribution. A minimal $k_{\mathrm{Cor}}$ in region II marks the extended regime.
% Each experimental data point is averaged over four repeated TOF images.
% Analysis of real- and momentum-space lengths identifies three regions (I–III), corresponding to localized, extended, and critical states. In region I, the real-space localization length is significantly smaller, consistent with localized states. In region II, the momentum-space correlation length is reduced, indicating extended states. Region III is characterized by large values of both length, consistent with critical states exhibiting multifractal wavefunctions. Experimental data and theoretical simulations show good agreement across all parameters.
Region III is characterized by large values of both length, consistent with critical states exhibiting multifractal wavefunctions. Experimental data and theoretical simulations show good agreement across different parameters and each experimental data point is averaged over four repeated TOF images.
}
    \label{fig:3}
\end{figure*}

The two-stage protocol enables a systematic experimental framework for the precise and selective preparation of states in the specific energy region. For fixed final Hamiltonian, the finally prepared quantum state of the ultracold atoms depends on the initial parameters $(A_i, f_i)$ [Fig.~\ref{fig:3}a]. We can map the prepared quantum state as a function of $(A_i, f_i)$, rendering an effective {\em preparation phase diagram} shown in Fig.~\ref{fig:3}b-e, which also clearly implies the existence of anomalous mobility edges in the tripartite phase.
%To quantitatively characterize the three types of quantum states prepared via the two-stage protocol, we systematically map the final-state properties as a function of the initial parameters $(A_i, f_i)$. As illustrated in Fig.~\ref{fig:3}(\textbf{A}), although all experimental trajectories terminate at the same final Hamiltonian,
%In particular, the pathway dependence of the Floquet evolution allows us to steer the system into different energy sectors.
We elaborate below how to determine this effective phase diagram in experiment.

We obtain a comprehensive library of TOF momentum distributions by scanning $f_i$ and $A_i$. These distributions, each averaged over four experimental realizations, reveal three distinct regions I-III in the $(A_i, f_i)$ parameter space, corresponding to the different prepared states and determined by two complementary characteristic lengths extracted from TOF images. The first observable is the spatial coherence length
\begin{eqnarray}
L_{\mathrm{Coh}}=\frac{\int_0^{x_\mathrm{max}} \mathrm{d}x \, |S(x) x|}{\int_0^{x_\mathrm{max}} \mathrm{d}x \, |S(x)|},\ S(x)=\int d x^{\prime}G\left(x^{\prime}, x+x^{\prime}\right),
\end{eqnarray}
which characterizes the spatial correlation and the degree of localization of the atomic ensemble, with the first-order spatial correlation function $G\left(x^{\prime}, x+x^{\prime}\right)=\langle\hat{\Psi}^{\dagger}\left(x^{\prime}\right) \hat{\Psi}\left(x+x^{\prime}\right)\rangle$ related to momentum-space density distribution $n(k) =\int \mathrm{d}x \, e^{\mathrm{i}kx} S(x)$ from the Wiener-Khinchin theorem~\cite{Deissler2010NP}. The second one is the density-density correlation length in momentum-space $k_{\mathrm{Cor}}$, given from $k_{\mathrm{Cor}}=\int_0^{k_0} |C(k)k| \, \mathrm{d}k / \int_0^{k_0} |C(k)| \, \mathrm{d}k$, with the correlation function $C(k)=\int_{-2k_0}^{2k_0} n(k'+k)\, n(k') \, \mathrm{d}k'$ \cite{Fang2016MomentumCorr}.
%For the many-particle system, $n(k)$ represents the collective distribution of the entire atomic cloud, and $L_{\mathrm{Coh}}$ effectively captures the mean distance over which phase coherence persists across the ensemble.
In the single-particle regime, the length $L_{\mathrm{Coh}}$ (dual to $k_{\mathrm{Cor}}$) is finite for a localized state but diverges for extended and critical states. For the many-particle system with atoms occupying multiple states, the magnitude of $L_{\mathrm{Coh}}$ characterizes the mean distance over which phase coherence persists across the ensemble, and is significantly reduced. Nevertheless, we find that the characteristic lengths still exhibit sharp difference between localized and delocalized regimes, %For a localized system, the restricted spatial support of the wave functions leads to a rapid decay of $S(x)$, resulting in a small $L_{\mathrm{Coh}}$. In contrast, for extended or critical states, the long-range spatial correlations manifest as a significantly larger $L_{\mathrm{Coh}}$.
providing a robust observable to determine localization-to-delocalization transitions \cite{Supplementary}.

In the experiment, the integration is bounded by $x_\mathrm{max}=5d$ (where $d=\SI{532}{nm}$ is the primary lattice constant) due to the length of the harmonic confinement. %Theoretically, $L_{\mathrm{Coh}}$ is finite for a localized state but diverges for extended and critical states. In real experiment the measured $L_{\mathrm{Coh}}$ is finite for extended and critical states, while distinguishable from that of a localized state \cite{Supplementary}.
As shown in Fig.~\ref{fig:3}b, a sharp contrast emerges at $A_i \approx \SI{42}{nm}$: Region I exhibits the minimal $L_{\mathrm{Coh}}$ in comparison with remaining regions, manifesting the localized states.
Extended states are characterized by sharp Bragg peaks, resulting in the minimal $k_{\mathrm{Cor}}$, whereas localized and critical states exhibit broader distributions. %Our results [Fig.~\ref{fig:3}(\textbf{D,E})] show that
With this we identify the Region II as the extended one with the uniquely small $k_{\mathrm{Cor}}$ [Fig.~\ref{fig:3}d]. The remaining Region III displays relatively large values for both $L_{\mathrm{Coh}}$ and $k_{\mathrm{Cor}}$. This "dual-delocalization" is a hallmark of critical states with a multifractal structure that is delocalized in both real and momentum spaces \cite{Gonfmmode2023Renorm}. The experimental results across the parameter space agree well with the numerical simulations [Fig.~\ref{fig:3}c,e]. The observations in Regions I-III are highly consistent with the preparation pathways and state classifications discussed above, confirming the high-feasibility preparation of the tripartite phase. The emergence of the three regions in the effective preparation phase diagram also implies existence of anomalous mobility edges in the spectrum.
% These measurements thus provide direct access to characteristic length scales of the final states, offering a clear way to distinguish localized, extended, and critical regimes in a driven lattice, and establishing a route to probe mobility edges with ultracold atoms.

\begin{figure*}
    \centering
    \includegraphics[width=1\textwidth]{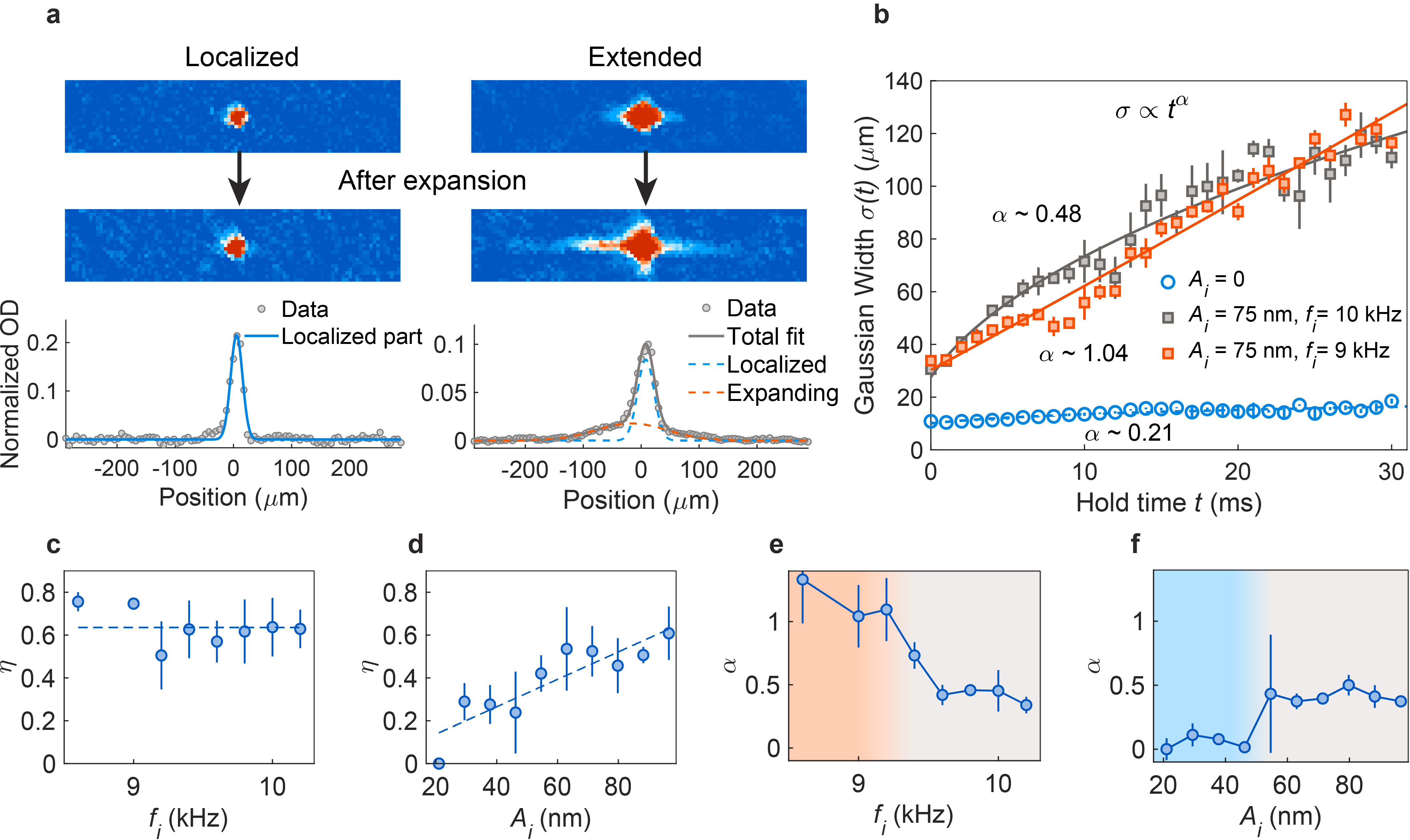}
    \caption{\textbf{Expansion dynamics.}
    \textbf{a}, Bimodal Gaussian fit of the expanded atomic cloud. The upper panel shows the in-situ absorption image, while the lower panels present Gaussian fits to the corresponding one-dimensional density profiles. In the localized regime, the atomic cloud remains nearly stationary after sudden trap release and a \SI{15}{ms} hold; a single Gaussian fit yields the cloud radius. In extended and critical regimes, significant expansion leads to a bimodal distribution. A bimodal Gaussian fit then determines the radii of both the expanding and localized components.
    \textbf{b}, Time evolution of the radii of three types of states. Radii are fitted to $\sigma=\sigma_0(1+t/t_0)^\alpha$ \cite{shimasaki2024anomalous} to extract the dynamical index $\alpha$. At $f_i=\SI{10}{kHz}$ (critical states), the expanding part shows $\alpha\approx 0.5$, while $\alpha\approx 1$ at $f_i=\SI{9}{kHz}$ (extended states), and $\alpha\approx 0.2$ when $A_i=0$ (localized states), consistent with theoretical predictions.
    \textbf{c},\textbf{d}, Fraction of the expanding component $\eta$ versus initial shaking frequency and amplitude. At fixed amplitude $A_i=\SI{75}{nm}$ (\textbf{c}), the fraction stays above 50\% with little variation in $f_i$. At fixed frequency $f_i=\SI{9}{kHz}$ (\textbf{d}), the fraction increases with amplitude, indicating that stronger shaking transfers more atoms to expanding states.
    \textbf{e},\textbf{f}, Scaling exponent $\alpha$ versus $f_i$ and $A_i$. In (\textbf{e}), $\alpha$ decreases sharply from ~1 to 0.5 near $f_i=\SI{9.4}{kHz}$, marking the transition from extended to critical states. In (\textbf{f}), $\alpha$ rises abruptly from 0 to ~0.5 near $A_i=\SI{50}{nm}$, indicating the transition from localized to critical states. Radii are determined adaptively: when the expanding fraction exceeds 30\%, the expanding part is used; otherwise, the localized part is taken.}
    \label{fig:4}
\end{figure*}

\section{Measuring the expansion exponents via transport dynamics}

The tripartite phase can be further directly observed by measuring the multiple universal expansion dynamics of the prepared wave-packets. When atoms are released into the $s$-$p$ orbital quasiperiodic lattice, their spatial width $\sigma(t)$ follows a universal power-law growth, $\sigma(t) \sim t^\alpha$, where the scaling exponent $\alpha$ serves as a dynamical signature of the underlying quantum states \cite{Hiramoto1988Dynamics,Zhong2001PRL,Wang2020PRL}. Theoretically, the three distinct types of states are characterized by specific exponents: $\alpha = 1$ for ballistic expansion (extended states), $\alpha = 0$ for asymptotic localization (localized states), and $\alpha \approx 0.5$ for anomalous diffusion (critical states).

The expansion measurements are initiated by abruptly switching off the optical dipole trap and the magnetic gradient field after the two-stage preparation sequence, while a sheet-light potential is switched on to levitate the atoms against gravity. %Driven by their fundamentally different transport natures, the three types of states can be distinguished after expansion. Experimentally, we observe three distinct dynamical behaviors as shown in Fig.~\ref{fig:4}(\textbf{B}), where the blue, red, and gray curves correspond to the localized, extended, and critical behaviors, respectively.
To quantitatively extract the scaling exponents of the expansion dynamics, we employ a bimodal Gaussian decomposition to model the evolving density profiles, as illustrated in Fig.~\ref{fig:4}a. Although our two-stage protocol is designed to selectively target specific energy sectors (e.g., critical states), a residual fraction of atoms inevitably persists in the initial localized $s$-orbital state due to finite inter-orbital coupling efficiency and inherent relaxation. Consequently, the expanding cloud manifests a characteristic bimodal spatial structure: a rapidly spreading component superposed on a stationary localized core. As shown in the \textit{in-situ} density distributions $n(x)$, the extended states exhibit a clear bimodal profile after $t = \SI{15}{ms}$. We therefore independently fit the temporal evolution of the widths for these distinct components using the power-law form  (Methods)~\cite{Supplementary}:
\begin{eqnarray}
\sigma = \sigma_0 (1 + t/t_0)^\alpha.
\end{eqnarray}
For the localized component, the extracted exponent $\alpha$ remains near zero. In contrast, for the expanding component, we resolve two distinct signatures: $\alpha \approx 1$ and $\alpha \approx 0.5$, corresponding to the successful preparation of extended and critical states, respectively [see Fig.~\ref{fig:4}b].

The population distribution between the localized and expanding components is further analyzed. The fraction of the expanding component $\eta$ is plotted against the preparation parameters $f_i$ and $A_i$ in Fig.~\ref{fig:4}c,d. Our data reveal that $\eta$ remains relatively insensitive to the initial frequency $f_i$ but increases significantly with the shaking amplitude $A_i$. This confirms that a stronger driving amplitude enhances the coupling to the $p$-orbital manifold, thereby transferring a larger fraction of atoms into the delocalized states.

To establish the universality of these transport signatures, we systematically scan the preparation parameters and extract the resulting expansion exponents, as shown in Fig.~\ref{fig:4}e,f. By varying $f_i$, we observe that $\alpha$ exhibits a sharp crossover from $\sim 1$ to $\sim 0.5$ at $f_i \approx \SI{9.4}{kHz}$, marking the transition from the extended to the critical states. Similarly, a scan of the amplitude $A_i$ reveals a clear transition from $\alpha \approx 0$ to $\alpha \approx 0.5$ as the shaking amplitude increases, identifying the transition between localized and critical states. These dynamical measurements align with our earlier state-preparation discussions and provide rigorous experimental verification of the prepared states. Collectively, these observations offer a systematic demonstration of the distinct expansion exponents for localized, extended, and critical states in this cold-atom platform, confirming the successful realization and observation of rich physics of the tripartite quantum phase.

\section{Conclusion and discussion}
We have reported the first experimental realization of the tripartite quantum phase in a quasiperiodic orbital optical lattice, featuring the coexistence of extended, localized, and critical states, which is not accessible in randomly disordered systems. Through a novel two-stage Floquet protocol, we achieve the precise preparation and characterization of these distinct types of states across the spectrum. The universal transport dynamics for the three types of states are observed, further confirming the realized tripartite quantum phase.

This work opens up a broad avenue to explore the previously inaccessible localization physics based on
quasiperiodic shaken lattices. Our Floquet modulation scheme not only establishes a versatile platform that can be applied to synthesize broad range of nontrivial quasiperiodic models, %beyond the current setup,
but also circumvents the state-preparation challenges inherent to static systems, offering a powerful toolbox for the precise realization and detection. In the noninteracting regime, it can accommodates internal degrees of freedom such as spin and orbital and can be extended to higher dimensions, thereby providing a clear pathway toward the experimental observation of all seven fundamental localization phases~\cite{Zhou2026}.  Furthermore, the present realization with ultracold atoms can be naturally applied to study interacting regimes, with the ability to selectively populate many-body eigenstates in targeted energy regions. %This capability lays the foundation to investigate the localization and critical phenomena under interactions, particularly the profound many-body critical phase~\cite{Wang2021PRL}, and provides the unprecedented opportunities to probe various many-body mobility edges~\cite{refs1,refs2} and interaction-driven phase transitions~\cite{refs3}, which are hitherto the major challenges in this direction, now becoming promising based on the present framework.
This capability lays the foundation to investigate interaction-driven localization and critical phenomena, including the profound many-body critical phase~\cite{Wang2021PRL}. It also offers unprecedented opportunities to probe the elusive many-body mobility edges~\cite{Laumann2014PRL,Naldesi2016SciPost,Roeck2016Absence,Li2015} and interaction-induced phase transitions~\cite{Yu2024BG}. These crucial and challenging issues are now within reach under the present framework.
%\vspace{0.3cm}

%\clearpage
\section*{Methods}
%\subsection*{Experimental setup}
\subsection{Orbital Quasiperiodic Lattice}
In this section, we derive the effective tight-binding model that captures the essential physics of the tripartite quantum phase. We achieve this by starting from a continuous-space description of the experimental system and projecting its dynamics onto a localized two-orbital basis. Experimentally, our setup consists of ultracold atoms confined in a one-dimensional quasiperiodic optical lattice, which is formed by superimposing a deep primary lattice with a weaker, incommensurate secondary lattice. To induce resonant transitions between the lowest two energy bands ($s$ and $p$ orbitals), the secondary lattice is periodically driven, or "shaken", through a time-dependent spatial modulation~\cite{Songbo2022}. In the laboratory frame, the dynamics of this driven quasiperiodic system are governed by the following effective Hamiltonian~\cite{Sakurai_Napolitano_2020}:
\begin{eqnarray}\label{eq: Hc}
%\begin{split}
	\hat{H}(t)&=&\int \mathrm{d}x\, \hat{\psi}^{\dagger}(x)\Big[
	-\frac{\hbar^2}{2m}\partial_x^2
	-V_0\sin^2(k_0x) \nonumber\\
	&&+V_1\sin^2\!\big(k_1(x-s(t))+\phi\big)
	\Big]\hat{\psi}(x),
%\end{split}
\end{eqnarray}
where $\hat{\psi}(x)$ is the bosonic field operator, $m$ the atomic mass, and $V_0$ ($V_1$) the depths of the primary (secondary) lattices. The secondary lattice displacement $s(t)=A\sin(\omega t)$ resonantly couples the lowest two energy bands~\cite{Songbo2022}.

To obtain a lattice model, we expand the field operator in terms of localized Wannier functions associated with the $s$ and $p$ orbitals of the primary lattice potential,
\begin{equation}
	\hat{\psi}(x)=\sum_j \Big[w_{s}(x-x_j)\,\hat a_j
	+w_{p}(x-x_j)\,\hat b_j\Big],
\end{equation}
where we have expanded the bosonic field operator in the localized Wannier basis of the lowest two orbitals ($s$ and $p$) centered at the lattice sites $x_j$. Here $\hat a_j$ and $\hat b_j$ are the annihilation operators for a particle occupying the $s$- and $p$-orbital Wannier states $w_s(x-x_j)$ and $w_p(x-x_j)$, respectively. Substituting this expansion into Eq.~\eqref{eq: Hc} and projecting onto the two-band Wannier subspace yields an effective tight-binding description. In carrying out this projection we retain the leading single-particle terms that are relevant for the physics discussed in the main text: the orbital-resolved on-site energies $\epsilon_{s,p}$, the quasiperiodic on-site modulations $\delta_{s,p}\cos(2\pi\alpha j)$ arising from the secondary lattice, the interorbital coupling terms (a static component $\delta_{sp}\cos(2\pi\alpha j)$ and a time-dependent component $\eta_{sp}\sin(2\pi\alpha j)\sin(\omega t)$ originating from the resonant lattice shaking), and the dominant nearest-neighbor hoppings $J_s$ and $J_p$ within each orbital. Terms of higher order in the Wannier overlap (e.g. longer-range hopping, small off-diagonal hoppings between different orbitals on neighboring sites, and interaction terms) are here neglected because the primary lattice is sufficiently deep that Wannier functions are strongly localized and nearest-neighbor processes dominate. This projection therefore yields the following reduced Hamiltonian (retaining the dominant single-particle contributions) which captures the resonant $s$–$p$ physics induced by the shaken quasiperiodic potential.
\begin{eqnarray}
	\hat H(t)&=&\sum_j \Big[(\epsilon_s+\delta_s \cos(2\pi \alpha j)\Big] \hat a_j^\dagger\hat a_j\nonumber\\
	&+&\sum_j \Big[(\epsilon_p+\delta_p \cos(2\pi\alpha j)\Big] \hat b_j^\dagger\hat b_j\nonumber\\
	&+&\sum_j \Big\{\bigr[(\delta_{sp}\cos(2\pi\alpha j)
	+\eta_{sp}\sin(2\pi \alpha j)\sin(\omega t)\bigr]\,\hat a_j^\dagger\hat b_j
	+\nonumber\\
	&-&J_s\hat a_j^\dagger\hat a_{j+1}
	+J_p\hat b_j^\dagger\hat b_{j+1}+\mathrm{H.c.}\Big\},
	\end{eqnarray}
where $\epsilon_s$ and $\epsilon_p$ denote the average on-site energies, $J_s$ and $J_p$ are the nearest-neighbor hopping amplitudes, $\delta_s$ and $\delta_p$ are the quasiperiodic onsite potentials, and $\eta_{sp}$ encodes the quasiperiodic interband coupling.

To eliminate the explicit time dependence and capture the resonant coupling between the $s$ and $p$ orbitals, we perform a rotating-wave transformation. Applying the unitary transformation $\hat{R}=\exp(i \omega t \sum_j \hat{b}_j^{\dagger} \hat{b}_j)$, the Hamiltonian in the rotating frame becomes $\hat{H}_{\mathrm{rot}}=\hat{R}[\hat{H}(t)-i\hbar \partial_t]\hat{R}^\dagger$. Retaining only the time-independent (zeroth-order) terms yields the effective static Hamiltonian:
\begin{eqnarray}
		\hat H_{\mathrm{eff}}&=&\sum_j \Big[\delta_s \cos(2\pi \alpha j)\, \hat a_j^\dagger\hat a_j
		+\big(\Delta +\delta_p \cos(2\pi\alpha j)\big)\hat b_j^\dagger\hat b_j\Big]\nonumber\\
		&+&\sum_j \Big(\tfrac{1}{2}\mathrm{i}\eta_{sp} \sin(2\pi \alpha j)\, \hat a_j^\dagger\hat b_j
		+\mathrm{H.c.}\Big)\nonumber\\
		&+&\sum_j \Big[-J_s\hat a_j^\dagger\hat a_{j+1}
		+J_p \hat b_j^\dagger\hat b_{j+1}+\mathrm{H.c.}\Big],
	\end{eqnarray}
	\label{eq:H}
where $\Delta=\epsilon_p-\epsilon_s-\hbar \omega$ is the effective detuning between the $p$ and $s$ orbitals\cite{Eckardt2017Colloquium,Bukov2015high-frequency}.
For realistic experimental parameters, such as a shaking amplitude of $A = \SI{42}{nm}$, the relevant model parameters are derived from the band dispersion and Wannier functions obtained via the Mathieu equations \cite{Drese1997Mathieu}. The resulting values are $\Delta_E = \epsilon_p - \epsilon_s = \SI{10.2}{kHz}$, $J_s = \SI{0.039}{kHz}$, $J_p = \SI{0.497}{kHz}$, $\delta_s = \SI{1.3}{kHz}$, $\delta_p = -\SI{0.8}{kHz}$, and $\eta_{sp} = \SI{1.13}{kHz}$. Notably, since this estimation neglects the deformation of the primary lattice Wannier functions induced by the secondary lattice, the actual quasiperiodic modulation strengths $\delta_s$ and $\delta_p$ are expected to be smaller than these calculated values. Defining $\Omega \equiv \eta_{sp}/2 = \SI{0.56}{kHz}$, the dimensionless parameters in units of $J_p$ are given by $\delta_p/J_p \approx 1.6$ and $\Omega/J_p \approx 1.1$. Under these conditions, the ratio $\delta_p/J_p < 2$ indicates that atoms in the $p$ orbital reside in the delocalized regime. In contrast, the $s$-band atoms are already deeply localized due to the suppressed tunneling $J_s \ll J_p$. Since the inter-orbital coupling $\Omega$ is comparable to $J_p$, it plays a decisive role in reshaping the localization landscape. As confirmed by the analytical and numerical investigations presented below, the system hosts a tripartite quantum phase consisting of localized $s$-orbital states, extended $p$-orbital states, and critical states emerging from the nontrivial $s$-$p$ coupling.

\subsection{Renormalization-group analysis of critical states}

To rigorously establish the existence of critical states in the two-orbital quasiperiodic model, we apply a renormalization-group (RG) method based on iteration of commensurate approximations~\cite{Miguel2022}. In this framework, the irrational frequency $\alpha$ is systematically approximated by a sequence of rationals $\alpha^{(n)}=p_n/q_n$, each defining a periodic Hamiltonian $H^{(n)}$ with supercell size $L=q_n$. The twisted boundary condition and the phase offset then serve as two quasi-momenta $\kappa_x$ and $\kappa_y$, mapping the problem onto a two-dimensional Bloch Hamiltonian. The RG flow as $n\to\infty$ dictates the asymptotic localization properties of the quasiperiodic eigenstates.

For each approximant, we evaluate the characteristic determinant $P^{(L)}(E;\kappa_x,\kappa_y)=\det(\mathcal{H}^{(L)}-E)$ and expand it in harmonics of the quasi-momenta. The leading Fourier coefficients define three renormalized parameters: $t_R^{(L)}$ (hopping along $\kappa_x$), $V_R^{(L)}$ (hopping along $\kappa_y$), and $C_R^{(L)}$ (diagonal coupling). To simplify the analysis, we work in the experimentally relevant regime $J_s\to 0$, where the $s$-orbital is frozen and the model reduces to an effective single-component quasiperiodic chain. In this limit, the $\kappa_x$-dependent terms arise solely from products of the $p$-orbital hopping $J_p$, and the $\kappa_y$-dependent terms are dominated by products of the $2\times 2$ on-site block determinants; higher-order harmonics become irrelevant under the RG flow. The phase is determined by the relative scaling of the three renormalized parameters: extended states correspond to dominant $t_R^{(L)}$, localized states to dominant $V_R^{(L)}$, and critical states emerge when all three terms remain comparably relevant.

Using a Chebyshev-polynomial identity for products of cosines over incommensurate phases, we evaluate these coefficients in closed form. The key results at energy $E=0$ are
\begin{equation}
C_R^{(L)}(0) \sim \Bigl(\tfrac{J_p\,\delta_s}{2}\Bigr)^{\!L},\qquad
\frac{C_R^{(L)}(0)}{V_R^{(L)}(0)} \sim \Bigl(\frac{2\,J_p\,\delta_s}{\delta_s\delta_p+\Omega^2}\Bigr)^{\!L},
\end{equation}
while $C_R^{(L)}(0)/t_R^{(L)}(0)=1$ holds exactly. Critical states at zero energy therefore persist whenever $C_R^{(L)}$ and $V_R^{(L)}$ scale comparably, which requires $2J_p\delta_s\geq\delta_s\delta_p+\Omega^2$, yielding the condition $\Omega^2+\delta_s\delta_p\leq 2\delta_s J_p$ stated in the main text. This analytically derived criterion is satisfied under the experimental parameters and provides a rigorous foundation for the tripartite phase. Detailed derivations are provided in the Supplementary Information~\cite{Supplementary}.

\subsection{Experimental Protocol}
A Bose–Einstein condensate of \num{1e5} atoms of \ce{^{87}Rb} in the $|F = 2, m_F = 2\rangle$ state is adiabatically loaded into a bichromatic lattice. Here $F$ denotes the total angular momentum and $m_F$ the magnetic quantum number of the state. All lattice depths are exponentially increased from zero to their final values in \SI{80}{ms} with a time constant of \SI{20}{ms}. After a hold time of \SI{13}{ms}, we begin to modulate (shake) the position of the secondary lattice by sinusoidally changing the corresponding laser frequency. To this end, we modulate the radio frequency(RF) feeding an acousto-optic modulator (AOM, MT110-B50A1-IR) in a double-pass configuration\cite{Songbo2022}. The resulting frequency shift of the lattice laser $\nu_s(t)$ becomes
\begin{equation}
	\nu_s(t)=\nu_c+\delta \nu_s \sin(2\pi f t),
\end{equation}
where $\nu_c=\SI{110}{MHz}$ is the center frequency of the AOM, $\delta \nu_s$ is the modulation depth, and $f$ denotes the modulation frequency. For a fixed distance $l_0=\SI{0.9}{m}$ between the atoms and the retro-reflection mirror, this frequency modulation leads to a time-dependent displacement of the lattice $s(t) = A \sin(2\pi f t)$, where the shaking amplitude $A$ denotes the maximum spatial displacement of the lattice. The amplitude $A$ is proportional to $\delta \nu_s$, and its absolute value is calibrated by the geometric relation between the laser wavelength, the retro-mirror distance $l_0$, and the AOM frequency change. The lattice constant is $d=\lambda_0/2=\SI{532}{nm}$, the waist of the primary lattice beam is \SI{100}{\micro m}, and the waist of the secondary lattice beam is \SI{150}{\micro m}. The momentum distribution is measured via time-of-flight absorption imaging.

In practice, we employ a Keysight 33600B waveform generator to control the frequency modulation input of a Rigol DG922 Pro, which provides the RF signal to the AOM. In this way, the Keysight 33600B controls the modulation of the AOM frequency and thus the modulation of the lattice position.

In our two-step preparation sequence, we first ramp up the shaking amplitude from 0 to $A_i$ within $t_0=\SI{1}{ms}$ while keeping the frequency fixed at $f_i$. We then linearly ramp the frequency from $f_i$ to $f_f=\SI{10.2}{kHz}$ up to $t_1=\SI{5}{ms}$, followed by a linear ramp of the amplitude from $A_i$ to $A_f=\SI{15}{nm}$ until $t_2=\SI{7}{ms}$. The final time of the shaking sequence is denoted as $t_3$. The time dependence of the amplitude is
\begin{equation}
	\begin{aligned}
		&A(t)= \begin{cases}
			A_i \frac{t}{t_1}, & 0 \leq t < t_1, \\
			A_i, & t_1 \leq t < t_2, \\
			\left(A_f-A_i\right) \frac{t-t_2}{t_3-t_2}+A_i, & t_2 \leq t \leq t_3,
		\end{cases}
	\end{aligned}
\end{equation}
while the frequency evolves as
\begin{equation}
	\begin{aligned}
		&f(t)= \begin{cases}
			f_i, & 0 \leq t < t_1, \\
			\left(f_f-f_i\right) \frac{t-t_1}{t_2-t_1}+f_i, & t_1 \leq t < t_2, \\
			f_f, & t_2 \leq t \leq t_3.
		\end{cases}
	\end{aligned}
\end{equation}
The timing sequence is illustrated in Supplementary Figure.~S2.

\subsection{Expansion Dynamics and Radius Extraction}

After preparing the final state, we measure the expansion dynamics by simultaneously switching off the optical trap and the magnetic trap that compensates gravity, thereby releasing the atoms from the external harmonic confinement. A levitating sheet beam is turned on at the same time to compensate gravity during free expansion. This beam is nearly collinear with the primary lattice (relative angle $\approx 10^{\circ}$), with a horizontal waist $w_h = 400{\mu m}$ and a vertical waist $w_v = 60{\mu m}$; the large horizontal extent ensures that atoms can expand freely along the lattice direction without being affected by the intensity profile of the levitating beam.

We observe that the atomic density distribution after a finite expansion time exhibits a bimodal profile, with a fraction of atoms expanding rapidly and the remainder staying localized. This bimodality arises because not all atoms are adiabatically transferred into the target state during the preparation sequence: a portion remains in the initial $s$-band localized state and is not excited into the $p$-band coupled regime. To quantitatively extract the expansion radius, we fit the measured density distribution with a double-Gaussian function,
\begin{equation}
\begin{split}
n(x) = \,& A_\mathrm{Loc}\exp\!\left[-\frac{(x-x_\mathrm{Loc})^2}{2\sigma_\mathrm{Loc}^2}\right] \\
&+ A_\mathrm{Exp}\exp\!\left[-\frac{(x-x_\mathrm{Exp})^2}{2\sigma_\mathrm{Exp}^2}\right],
\end{split}
\end{equation}
where $\sigma_\mathrm{Loc}$ and $\sigma_\mathrm{Exp}$ denote the widths of the localized and expanding components, respectively, and $A_\mathrm{Loc}$, $A_\mathrm{Exp}$ are their corresponding amplitudes. The fraction of the expanding component is defined as
\begin{equation}
\eta = \frac{A_\mathrm{Exp}\,\sigma_\mathrm{Exp}}{A_\mathrm{Loc}\,\sigma_\mathrm{Loc} + A_\mathrm{Exp}\,\sigma_\mathrm{Exp}}.
\end{equation}
When the final state is extended, the expansion radius $\sigma_\mathrm{Exp}$ is directly extracted from the fit. When the final state is localized, no significant expanding component is present, and the width of the localized component $\sigma_\mathrm{Loc}$ is used instead. To ensure consistency across different preparation parameters, we apply a threshold criterion: when $\eta \leq 30\%$, the reported radius is taken as $\sigma_\mathrm{Loc}$; otherwise it is taken as $\sigma_\mathrm{Exp}$. To ensure consistency across different preparation parameters, we apply a threshold
criterion: $\sigma = \sigma_\mathrm{Loc}$ when $\eta \leq 30\%$,
and $\sigma = \sigma_\mathrm{Exp}$ when $\eta > 30\%$,
ensuring that the extracted radius consistently reflects the physically relevant contribution in each regime.

%%%%%%%%%%%%%%%%%%%%%%%%%%%%%%%%%%%%%

\normalem
\bibliographystyle{naturemag}
\nolinenumbers
\bibliography{ref}

\paragraph*{\bf{Acknowledgement}}
We thank Hepeng Yao for helpful discussions. This work was supported by the National Key Research and Development Program of China (2021YFA1400900), National Natural Science Foundation of China (No.~12425401 and No.~12261160368), Quantum Science and Technology-National Science and Technology Major Project (2021ZD0302000), and Shanghai Municipal Science and Technology Major Project (2019SHZDZX01). X.-J.L. was also supported by the New Cornerstone Science Foundation through the XPLORER PRIZE.

\paragraph*{\bf{Author Contributions}}
X.-J. L. conceived the research and designed the experimental scheme. Z. H., Y. G., Y.-D. W., Z. Q., J. Y., X. C., and S. J. set up the experiment. Z. H., Y. G., and S. J. performed the measurements. Z. H., Y. G., and Y.-D. W. analyzed the data. Z. H., B.-C. Y., X.-C. Z., B.-Z. W., Y. G., Z. Q. provided the theoretical calculations. S. J. and X.-J. L. supervised the research. Z. H., Y. G., S. J., and X.-J. L. wrote the paper.  All authors read, edited, and approved the final manuscript.

%\paragraph*{\bf{Data Availability}}

\clearpage
\onecolumngrid

% 重新开始 section 编号
\setcounter{section}{0}
\setcounter{subsection}{0}
\setcounter{figure}{0}
\setcounter{table}{0}
\setcounter{equation}{0}

\section*{Supplementary Material}

\section{Tripartite Quantum Phase}
\subsection{Analytical analysis}
In the following, we first perform a quantitative analysis based on the concept of incomemensurately distributed zeros (IDZs)~\cite{avila2015global,Zhou2023PRL,zhou2025fundamental}, which provides a mechanism for the emergence of a tripartite quantum phase featuring the coexistence of extended, localized, and critical states. After that, we provide a renormalization group treatment~\cite{Miguel2022,Gonfmmode2023Renorm}, which independently corroborates the existence of critical states at zero energy. 

% \subsubsection{Mechanism for the emergence of critical states}
\subsubsection{Mechanism for the tripartite phase}
As outlined in the main text, the coexistence of extended, localized, and critical states in the present model can be understood by analyzing eigenstates residing in distinct energy regimes. For convenience,we freeze the $s$ orbital by taking $J_s\to0$ and consider parameters satisfying $J_p\ge\delta_s\sim\delta_p>\Omega$. In this regime, spectrum with energies $|E|\le\delta_s$ displays emergent spectral singularities, leading to either localized or critical behavior depending on the eigen-energy, whereas eigenstates with $|E|>\delta_s$ are free of such singularities and remain extended, as we provide detailed analysis below. 

To facilitate the discussion, we rewrite the Hamiltonian Eq.~(8) in main text as
\begin{equation}
\hat H=\sum_j \left(\hat c^\dagger_{j+1} \Pi_j \hat c_j +\mathrm{H.c.}\right)+\hat c^\dagger_{j} M_j \hat c_j,
\end{equation}
with $\hat c^\dagger_j=(a^\dagger_j,b^\dagger_j)$, and the hopping coupling matrix and on-site matrix are given by
\begin{equation}\label{eq:Pauli-Coupling}
    \begin{aligned}
        \Pi_j=&
        \begin{pmatrix}
            -J_s&0\\
            0&J_p
        \end{pmatrix}=t_+\sigma_0+t_-\sigma_z,\\
        M_j=&
        \begin{pmatrix}
            \delta_s \cos(2\pi \alpha j)&\mathrm{i}\Omega \sin(2\pi \alpha j)\\
            -\mathrm{i}\Omega \sin(2\pi \alpha j)&\delta_p \cos(2\pi \alpha j)
        \end{pmatrix}=\bar\delta \cos(2\pi \alpha j)\sigma_0+\Omega \sin(2\pi \alpha j)\sigma_y+\delta' \cos(2\pi \alpha j)\sigma_z.
    \end{aligned}
\end{equation}
Here $t_\pm=(-J_s\pm J_p)/2,\, \bar \delta=(\delta_s+\delta_p)/2,\, \delta'=(\delta_s-\delta_p)/2$, and $\sigma_{0,x,y,z}$ are the identity and Pauli matrices. For experimentally relevant parameters, $J_s\ll J_p$, allowing us to approximate $J_s\simeq0$. Under this condition, the hopping coupling matrix satisfies $\det{\Pi_j}=0$, allowing the model to be effectively reduced to a spinless form. For an eigenstate $H|\psi\rangle=E|\psi\rangle$, we denote the wavefunction as $|\psi\rangle=\sum_j (u_j \hat b_j^\dagger + v_j \hat a_j^\dagger)|0\rangle$, where $u_j$ and $v_j$ denote amplitudes on the $p$ and $s$ orbitals, respectively. The corresponding coupled equations read
\begin{equation}
    \begin{aligned}
    J_p u_{j+1}+J_p u_{j-1}+\delta_p \cos (2 \pi \alpha j) u_j+\mathrm{i} \Omega \sin (2 \pi \alpha j) v_j&=E u_j, \\
    \delta_s \cos (2 \pi \alpha j) v_j-\mathrm{i} \Omega \sin (2 \pi \alpha j) u_j&=E v_j.
    \end{aligned}
\end{equation}

By algebraically eliminating $v_j$, corresponding physically to integrating out the frozen $s$ orbital, one obtains an effective single-component recursion relation for $u_j$. The resulting equation reads
% As a result, the dynamics of the system can be effectively captured by the evolution of the $p$-orbital amplitudes $u_j$, while the effect of the eliminated component is absorbed into an energy-dependent modification of the on-site potential term. This leads to the following effective single-component equation:
\begin{equation}
     J_p \left(u_{j-1}+u_{j+1}\right)+\left[\delta_p \cos (2 \pi \alpha j)+\frac{\Omega^2 \sin ^2(2 \pi \alpha j)}{E-\delta_s \cos (2 \pi \alpha j)}\right] u_j=E u_j.
\end{equation}
The system is thus effectively mapped onto a spinless quasiperiodic model with uniform hopping amplitude $J_p$. The energy-dependent on-site potential encodes the influence of $s$-$p$ inter-orbital coupling and is given by
\begin{equation}
V_{\mathrm{eff},j}=\delta_p \cos (2 \pi \alpha j)+\frac{\Omega^2 \sin ^2(2 \pi \alpha j)}{E-\delta_s \cos (2 \pi \alpha j)}.
\end{equation}
% The second term in $V_{\mathrm{eff},j}$ arises from the virtual hybridization between the $s$ and $p$ orbitals and introduces both energy and position dependence into the potential landscape, thereby enriching the localization properties and allowing for the emergence of critical states in the effective quasiperiodic model.

This expression immediately reveals a qualitative distinction between energy regimes. For $|E|>\delta_s$, the denominator remains finite for all $\tilde{j}$, ensuring that $V_{\mathrm{eff},\tilde{j}}$ is bounded and smooth. In contrast, for $|E|<\delta_s$, there exist incommensurate sites $j_0$ satisfying $E-\delta_s \cos (2 \pi \alpha j_0)=0$, at which $V_{\mathrm{eff},j}$ diverges.

The eigenstates close to the band edge, with $|E|>\delta_s$, are extended states. Since such energies lie outside the range of the $s$-orbital on-site energy, they are mainly contributed by the $p$-orbital component, where the hopping amplitude $J_p$ dominates over the quasiperiodic modulation, thus resulting in extended states.

For the eigenstates within the $s$ orbital onsite energy $|E|<\delta_s$, the eigenstates can be either localized or critical states, depending on the details of the state. The inter-orbital coupling in Eq.~(\ref{eq:Pauli-Coupling}) facilitates a complete transition from $p$ to $s$ orbitals at the $\tilde j$-sites where $\delta_{s,\tilde j}=E$, which are incommensurately distributed throughout the system. This full transition to the $s$-orbital hinders the tunneling of the $p$ orbitals across those sites, resulting in incommensurately distributed zeros (IDZs) in hopping—a central mechanism leading to the emergence of critical states, as captured by Avila's profound global theory~\cite{avila2015global,Zhou2023PRL}. 

With the presence of IDZs, two distinct scenarios follow. For states with $|E|\lesssim\delta_s$, eigenstates reside near extrema of the onsite potential $\delta_{i,j}$ and are significantly influenced by the quasiperiodic potential, resulting in localized states. In contrast, states with $E\sim0$ are primarily located at the nodes of the onsite potential $\delta_{s,j}\sim \cos(2\pi \alpha j)$, where the inter-orbital coupling $\sim \sin(2\pi \alpha j)$ is relatively strong. These states are less affected by the quasiperiodic potential, resulting in delocalization. Consequently, the IDZs effectively partition the lattice into segments, within which the delocalized wavefunction reorganizes self-similarly, giving rise to scale-free critical states.

The appearance of critical states near $E\sim0$ motivates a  renormalization-group analysis of the coupling flow at zero energy, which we present in the following section.

\subsubsection{Renormalization group analysis}

The statement that the eigenstates near $E \sim 0$ are critical can be further corroborated through a renormalization group analysis based on iteration of commensurate approximation~\cite{Miguel2022,Gonfmmode2023Renorm}. The renormalization group (RG) framework analyzes the relevance of effective hopping coefficients by iteratively applying rational approximations to the quasiperiodic (QP) parameter. Concretely, one considers a sequence of rationals $\alpha^{(n)} = p_n/q_n$ that define a family of periodic Hamiltonians $H^{(n)}$, which recover the true quasiperiodic regime only in the irrational limit $\alpha = \alpha^{(\infty)}$. The RG flow as $n \rightarrow \infty$ then dictates the asymptotic properties of the eigenstates in the QP system. For each rational approximation, the spectrum of $H^{(n)}$ forms a periodic function of the quasi-momenta $\kappa_x$ and $\kappa_y$. To systematically track how the renormalized coefficients evolve with system size, we evaluate the characteristic determinant
\begin{equation}
 P^{(n)} = \big|\mathcal{H}^{(n)} - E \big|
\end{equation}
 at a target energy $E$ for a system size $L=F_n$, which can be expressed as
\begin{equation}
 \begin{aligned}
 P^{(n)}(E; \kappa_x,\kappa_y) & = t_R^{(n)} \cos(\kappa_x + \kappa_x^0) + V_R^{(n)} \cos(\kappa_y + \kappa_y^0) \\
 & \quad + C_R^{(n)} \cos(\kappa_x + \tilde{\kappa}_x^0)\cos(\kappa_y + \tilde{\kappa}_y^0) \\
 & \quad + \epsilon_R^{(n)}(E,\varphi,\kappa) + T_R^{(n)}(E),
 \end{aligned}
\end{equation}
 where $\epsilon_R^{(n)}$ contains higher-order harmonics that become irrelevant under the RG flow.

The long-scale nature of the eigenstates is governed by the relative scaling of the renormalized hopping amplitudes. For extended states, hopping along the $x$ direction dominates, such that $\lvert C_R^{(n)}/t_R^{(n)} \rvert, \lvert V_R^{(n)}/t_R^{(n)} \rvert \rightarrow 0$. For localized states, hopping is primarily along $y$, giving $\lvert t_R^{(n)}/V_R^{(n)} \rvert, \lvert C_R^{(n)}/V_R^{(n)} \rvert \rightarrow 0$. Critical states emerge when all three hopping terms remain equally relevant, characterized by $\lvert C_R^{(n)}/V_R^{(n)} \rvert, \lvert C_R^{(n)}/t_R^{(n)} \rvert \gtrsim 1$. The RG flow thus provides a systematic means to identify transitions among extended, localized, and critical phases, as well as to locate mobility edges that appear as energy-dependent transition points.

To explicitly characterize the localized, critical, and extended regions in our incommensurate s-p orbital lattices, where we calculate the characteristic polynomial of commensurate approximates (CA) of the system at thermodynamic limit and compare different renormalized factors. The characteristic polynomial $P^{(L)}(E ; \varphi, \kappa)=\operatorname{det}(H(\varphi, \kappa)-E)$ for any CA at system length $L$, boundary phase twist $\kappa=L k$, accumulated initial phase $\varphi=L \phi$ and energy $E$ is calculated as

\begin{equation}
\begin{aligned}
P^{(L)}(E;\varphi,\kappa)
= \det\!\begin{pmatrix}
H_0 - E I_2 & T e^{i \kappa} & 0 & \cdots & 0 & T^\dagger e^{-i \kappa} \\
T^\dagger e^{-i \kappa}& H_1 - E I_2 & Te^{i \kappa} & \cdots & 0 & 0 \\
0 & T^\dagger e^{-i \kappa}& H_2 - E I_2 & \ddots & \vdots & \vdots \\
\vdots & \vdots & \ddots & \ddots & Te^{i \kappa} & 0 \\
0 & 0 & \cdots & T^\dagger e^{-i \kappa}& H_{L-2} - E I_2 & Te^{i \kappa} \\
T e^{i \kappa} & 0 & \cdots & 0 & T^\dagger e^{-i \kappa}& H_{L-1} - E I_2
\end{pmatrix},
\end{aligned}
\end{equation}
where 
\begin{equation}
T =
\begin{pmatrix}
0 & 0 \\
0 & J_p
\end{pmatrix},\quad
H_j =
\begin{pmatrix}
\delta_s \cos(2\pi\alpha j + \phi) & i \Omega \sin(2\pi\alpha j + \phi) \\
- i \Omega \sin(2\pi\alpha j + \phi) & \delta_p \cos(2\pi\alpha j + \phi)
\end{pmatrix},
\, j = 0,1,\ldots,L-1,
\end{equation}
and $I_2$ is the $2\times2$ identity matrix. This is a periodic function of $\kappa$ and $\varphi$. By expanding $P^{(L)}(E ; \varphi, \kappa)$ on different frequencies of $\kappa$ and $\varphi$, the factors before $e^{i \kappa}, e^{i \varphi}$ and $e^{i(\kappa+\varphi)}$, denoted as $t_R^{(L)}(E), V_R^{(L)}(E)$ and $C_R^{(L)}(E)$ respectively, determines whether a state is in localized, extended or critical phase. We analyze these factors by carefully counting the phase factors. To accumulate a $k$-dependence of $e^{i L k}$, the only choice is to multiple $L$ off-diagonal block $J_p$ terms together, thus the factors $C_R^{(L)}(E)$ and $t_R^{(L)}(E)$ are completely contributed from
\begin{equation}
J_p^L \Pi_{j=0}^{L-1}\left[\delta_s \cos (2 \pi j \alpha+\phi)-E\right]
\end{equation}
For a general $J_p$, the determination of $V_R^{(L)}(E)$ is complicated since after picking $L e^{i \phi}$ terms, the leftover $L$ terms are picked among $e^{i \phi}$ terms, $e^{i k}$ terms or constant phase dependency terms in various ways. However, for small $J_p$ compared to other parameters, the contribution from off-diagonal blocked can be safely ignored, leaving the coefficient determinable by multiplication of determinants of diagonal $2 \times 2$ blocks,

\begin{equation}
\Pi_{j=0}^{L-1}\left[\left(\delta_s \delta_p+\Omega^2\right) \cos ^2(2 \pi j \alpha+\phi)-E\left(\delta_s+\delta_p\right) \cos (2 \pi j \alpha+\phi)+\left(E^2-\Omega^2\right)\right]
\end{equation}
By evaluating these two expressions, we derive the desired renormalized factors. Calculation of all these factors thus break down to multiplications of the form $b \cos (2 \pi j \alpha+\phi)+a$, which we calculate exactly with the following lemma.

\begin{lemma}
For $\alpha=\frac{p}{L}$ with $p$ coprime to $L$, we have
\begin{equation}
\prod_{j=0}^{L-1}\!\bigl[b \cos (2 \pi j \alpha+\phi)+a\bigr]
= \left(-\tfrac{1}{2}\right)^{L-1} b^L \Bigl[ \cos (L \phi)-(-1)^L T_L\!\left(\tfrac{a}{b}\right)\Bigr],
\end{equation}
where $T_L$ is the $L$-th Chebyshev polynomial.
\end{lemma}

\begin{proof}
Take $z_j = e^{i(2 \pi j \alpha+\phi)}$ and $c=\tfrac{2 a}{b}$. The left-hand side can be rewritten as
\begin{align}
\prod_{j=0}^{L-1}[b \cos (2 \pi j \alpha+\phi)+a]
&= \left(\frac{b}{2}\right)^L \prod_{j=0}^{L-1}\!\left(z_j+z_j^{-1}+\frac{2 a}{b}\right) \notag = \left(\frac{b}{2}\right)^L \prod_{j=0}^{L-1}\!z_j^{-1}\,
\prod_{j=0}^{L-1}\!\left(z_j^2+c z_j+1\right).
\end{align}

The first product can be evaluated straightforwardly as
$
\prod_{j=0}^{L-1} z_j^{-1} = e^{-i L \phi}\, e^{-i \pi \alpha L(L-1)}.
$
For the second product, each quadratic factor can be decomposed as
$
z_j^2 + c z_j + 1 = (z_j - r_{+})(z_j - r_{-}),
$
where
$
r_{\pm} = \frac{-c \pm \sqrt{c^2 - 4}}{2}
$
are the two roots of the characteristic equation.  
To proceed, we introduce the $L$-th root of unity \( \omega = e^{i 2\pi \alpha} \), such that \( z_j = e^{i\phi} \omega^j \). Thus\begin{align}
\prod_{j=0}^{L-1}[b \cos (2 \pi j \alpha+\phi)+a]
&= \left(\frac{b}{2}\right)^L e^{i L \phi} \,\omega^{-\tfrac{L(L-1)}{2}}
\prod_{j=0}^{L-1}\!\bigl(-\omega^j+r_{+} e^{-i \phi}\bigr)
\prod_{j=0}^{L-1}\!\bigl(-\omega^j+r_{-} e^{-i \phi}\bigr) \notag \\
&= \left(\frac{b}{2}\right)^L e^{i L \phi} \,\omega^{-\tfrac{L(L-1)}{2}}
\bigl(r_{+}^L e^{-i L \phi}-1\bigr)\bigl(r_{-}^L e^{-i L \phi}-1\bigr) \notag \\
&= \left(\frac{b}{2}\right)^L (-1)^{p(L-1)} e^{i L \phi}
\left[ r_{+}^L r_{-}^L e^{-2 i L \phi} + 1 - (r_{+}^L+r_{-}^L) e^{-i L \phi}\right] \notag \\
&= \left(-\tfrac{1}{2}\right)^{L-1} b^L \left(\frac{e^{i L \phi}+e^{-i L \phi}}{2}
- \frac{r_{+}^L+r_{-}^L}{2}\right) \notag \\
&= \left(-\tfrac{1}{2}\right)^{L-1} b^L
\Bigl[\cos (L \phi)-(-1)^L T_L\!\left(\tfrac{a}{b}\right)\Bigr].
\end{align}
This proves the lemma.
\end{proof}

\bigskip

Using the lemma, the RG coefficients at $L \to \infty$ are obtained as
\begin{align}
t_R^{(L)}(E) &\sim \left(\tfrac{J_p \delta_s}{2}\right)^L T_L\!\left(\tfrac{E}{\delta_s}\right) \notag \sim 
\begin{cases}
\left(\tfrac{J_p \delta_s}{2}\right)^L, & |E|<|\delta_s|, \\[6pt]
\left[\tfrac{J_p\left(|E|+\sqrt{E^2-\delta_s^2}\right)}{2}\right]^L, & |E|>|\delta_s|,
\end{cases} 
\end{align}\label{eq:tR}

\begin{align}
V_R^{(L)}(E) 
\gtrsim &\left(\tfrac{\delta_s \delta_p+\Omega^2}{4}\right)^L
\Biggl[
T_L\!\left(\frac{E(\delta_s+\delta_p)+\sqrt{E^2(\delta_s+\delta_p)^2-4(E^2-\Omega^2)(\delta_s \delta_p+\Omega^2)}}{2(\delta_s \delta_p+\Omega^2)}\right) \notag \\
&\qquad\quad\quad\quad +\,
T_L\!\left(\frac{E(\delta_s+\delta_p)-\sqrt{E^2(\delta_s+\delta_p)^2-4(E^2-\Omega^2)(\delta_s \delta_p+\Omega^2)}}{2(\delta_s \delta_p+\Omega^2)}\right)
\Biggr] \notag \\
\sim& \left(\tfrac{\delta_s \delta_p+\Omega^2}{4}\right)^L 
\max \{ w_{\pm}^L, w_{\pm}^{-L} \}, \label{eq:VR}
\end{align}
where
\begin{equation}
w_{\pm} = \left| z_{\pm}+\sqrt{z_{\pm}^2-1} \right|,
\qquad
z_{\pm} = \frac{E(\delta_s+\delta_p) \pm \sqrt{E^2(\delta_s+\delta_p)^2-4(E^2-\Omega^2)(\delta_s \delta_p+\Omega^2)}}{2(\delta_s \delta_p+\Omega^2)}.
\end{equation}
Finally,
\begin{equation}
C_R^{(L)}(E) \sim \left(\tfrac{J_p \delta_s}{2}\right)^L.
\end{equation}
Where the approximate equality holds at small $J_p$ for the second equation, $w_{ \pm}=\left|z_{ \pm}+\sqrt{z_{ \pm}^2-1}\right|, z_{ \pm}=\frac{E\left(\delta_s+\delta_p\right) \pm \sqrt{E^2\left(\delta_s+\delta_p\right)^2-4\left(E^2-\Omega^2\right)\left(\delta_s \delta_p+\Omega^2\right)}}{2\left(\delta_s \delta_p+\Omega^2\right)}$. Coefficients of order $o\left(e^L\right)$ have been ignored at thermodynamic limit. Note the second equation is also exact at $E=0$ regardless of the value of $J_p$, since if $E=0$, the leftover terms can be picked only among $e^{i \phi}$ terms and $e^{i k}$ terms. Then take CAs with an odd $L$, for example take only odd terms in the Fibonacci sequence. Now $e^{i \phi}$ terms and $e^{i k}$ terms must contribute a total of zero phase altogether while there is an odd number of phase factors to pick, giving a contradiction, leading to no $e^{i \varphi}$ terms. Instead we must consider $e^{2 i \varphi}$ term, given as $V_{2 R}^{(L)}(E) \sim\left(\frac{\delta_s \delta_p+\Omega^2}{4}\right)^L$, which is exactly $\lim _{E \rightarrow 0} V_R^{(L)}(E)$. 
Consequently, the phase structure and phase transitions of the model are determined by this limiting behavior. At \( E = 0 \), one can rigorously compare the relative magnitudes of the three terms, yielding
\begin{equation}
\frac{C_R^{(L)}(E=0)}{t_R^{(L)}(E=0)} = 1, \qquad
\frac{C_R^{(L)}(E=0)}{V_R^{(L)}(E=0)} \sim 
\left(\frac{2 J_p \delta_s}{\delta_s \delta_p + \Omega^2}\right)^L.
\end{equation}
This result implies that critical states appear strictly at zero energy whenever the condition
$
\Omega^2 + \delta_s \delta_p \le 2 \delta_s J_p
$
is satisfied, which holds under the experimental parameters considered here.

\begin{figure}[htbp]
\centering
\includegraphics[width=\textwidth]{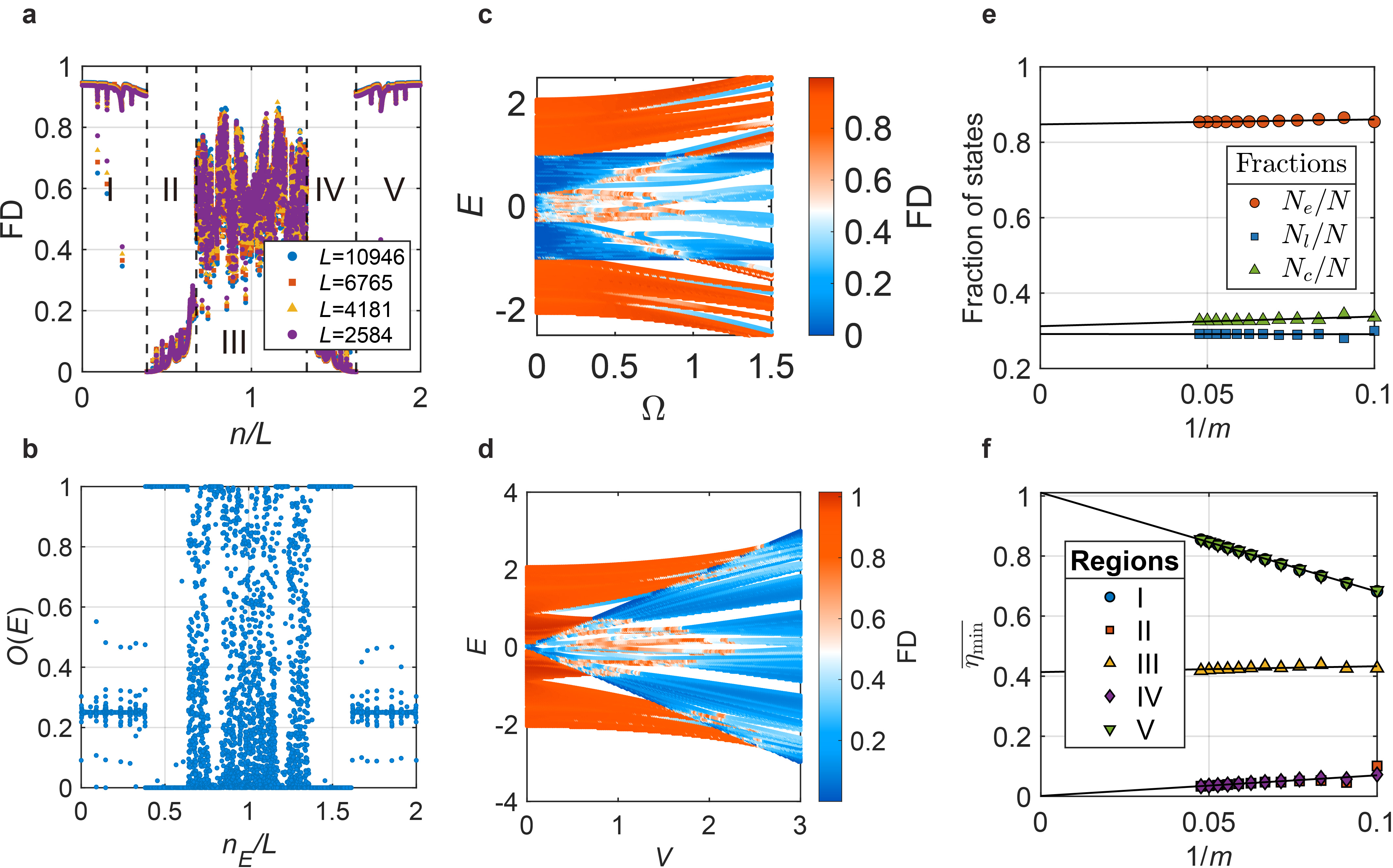}
\caption{Numerical results of the orbital quasiperiodic lattice model. \textbf{a}, Fractal dimension (FD) versus state index $n$ for different system sizes $L$, with parameters $V=J_p,\, \Omega=0.5J_p,\,\Delta=0,\, \alpha=(\sqrt{5}-1)/2$. Five distinct regions can be identified. \textbf{b}, Quarter-lattice occupation number $O(E)$\cite{Wang2022PRB} versus state index $n_E$ for $L=6765$, corresponding to the regions in (\textbf{a}). \textbf{c},\textbf{d}, FD as a function of coupling $\Omega$ and quasiperiodic strength $V$. The light red region represents the critical regime. \textbf{e},\textbf{f}, Finite-size scaling: the fractions of three types of states (\textbf{e}) and the scaling of $\overline{\eta_{\min}}$ (\textbf{f}) for various Fibonacci system sizes $L=F_m$.}
\label{fig:Supple_Numerical}
\end{figure}

\subsection{Numerical results}

To further validate the existence of the tripartite quantum phase, we perform comprehensive numerical simulations on the effective Hamiltonian $\hat{H}_{\mathrm{eff}}$ defined in Eq.~(8) (main text). Our investigation focuses on characterizing the localization properties of the entire energy spectrum through exact diagonalization (ED) and finite-size scaling analysis. By computing a suite of complementary diagnostics—including the fractal dimension (FD), the quarter-lattice occupation number $O(E)$, and the averaged minimal scaling exponent $\overline{\eta_{\min}}$—we provide a multi-dimensional perspective on the distinct transport and spatial properties of the localized, extended, and critical states.

Specifically, we first evaluate the FD for each eigenstate $|\psi_n\rangle$, defined as:
\begin{equation}
\mathrm{FD}=-\frac{\ln\left(\sum_j |u_{j,n}|^4+|v_{j,n}|^4\right)}{\ln L},
\end{equation}
where $u_{j,n}$ and $v_{j,n}$ are the amplitudes of the $n$-th eigenstate on the $s$ and $p$ orbitals at site $j$, respectively. The FD serves as a sensitive probe for the spatial structure of the wave function: in the thermodynamic limit, $\mathrm{FD} \to 1$ for delocalized (extended) states, $\mathrm{FD} \to 0$ for localized states, while a fractional value $0 < \mathrm{FD} < 1$ signifies the self-similar, multi-fractal nature characteristic of critical states. Additionally, we calculate the quarter-lattice occupation number $O(E) = \sum_{j=1}^{L/4} \langle \psi_E | \hat{n}_{j,s} + \hat{n}_{j,p} | \psi_E \rangle$ to assess the spatial distribution of the particle density across the lattice. For delocalized states, the density is expected to be uniformly distributed, yielding $O(E) \approx 1/4$, whereas localized states tend to cluster in specific regions, leading to $O(E)$ values near $0$ or $1$, depending on the localization center~\cite{Wang2022PRB}.

The numerical results presented in Fig.~\ref{fig:Supple_Numerical}a,b reveal a clear tripartite separation within the energy spectrum. Based on the calculated FD and $O(E)$, we identify five distinct spectral regions. In regions I and V, the eigenstates exhibit $\mathrm{FD} \approx 1$ and $O(E) \approx 1/4$, which are the hallmark signatures of extended states. Conversely, regions II and IV are characterized by $\mathrm{FD} \approx 0$ and $O(E)$ values approaching the boundaries of $0$ or $1$, indicating strong localization. Most importantly, region III displays intermediate FD values and significant fluctuations in $O(E)$, providing robust evidence for the emergence of critical states. This tripartite structure is further explored in Fig.~\ref{fig:Supple_Numerical}c,d, where we map the FD as a function of the interorbital coupling $\Omega$ and the quasiperiodic strength $V$. The resulting phase diagrams clearly delineate the boundaries of the critical regime, highlighting how the interplay between orbital hybridization and quasiperiodicity reshapes the localization landscape.

Finally, to ensure that these observations persist in the thermodynamic limit, we perform a rigorous finite-size scaling analysis. We select system sizes $L=F_m$ corresponding to the Fibonacci sequence to maintain the self-similarity of the quasiperiodic potential. As shown in Fig.~\ref{fig:Supple_Numerical}e,f, we track the fraction of each type of state and the averaged minimal exponent $\overline{\eta_{\min}}$, defined as:
\begin{equation}\overline{\eta_{\min }}=\frac{1}{N_s} \sum_{\text {states in region}} \eta_{\min}, \quad \text{with} \quad\eta_{\min }=\min_j\left\{-\frac{\ln (n_{j,s} + n_{j,p})}{\ln L}\right\}.
\end{equation}
The scaling exponent $\eta_{\min}$ is a powerful tool for distinguishing localization properties: as $L \to \infty$, it converges to $1$ for extended states, $0$ for localized states, and a value between $0$ and $1$ for critical states. Our scaling results demonstrate that $\overline{\eta_{\min}}$ for each identified region converges toward these distinct asymptotic limits, confirming that the tripartite quantum phase is a robust feature of the system rather than a finite-size effect~\cite{Wang2020PRL_ME,Xiaopeng2016,Zhou2023PRL}.

\begin{figure}[htbp]
\centering
\includegraphics[width=0.6\textwidth]{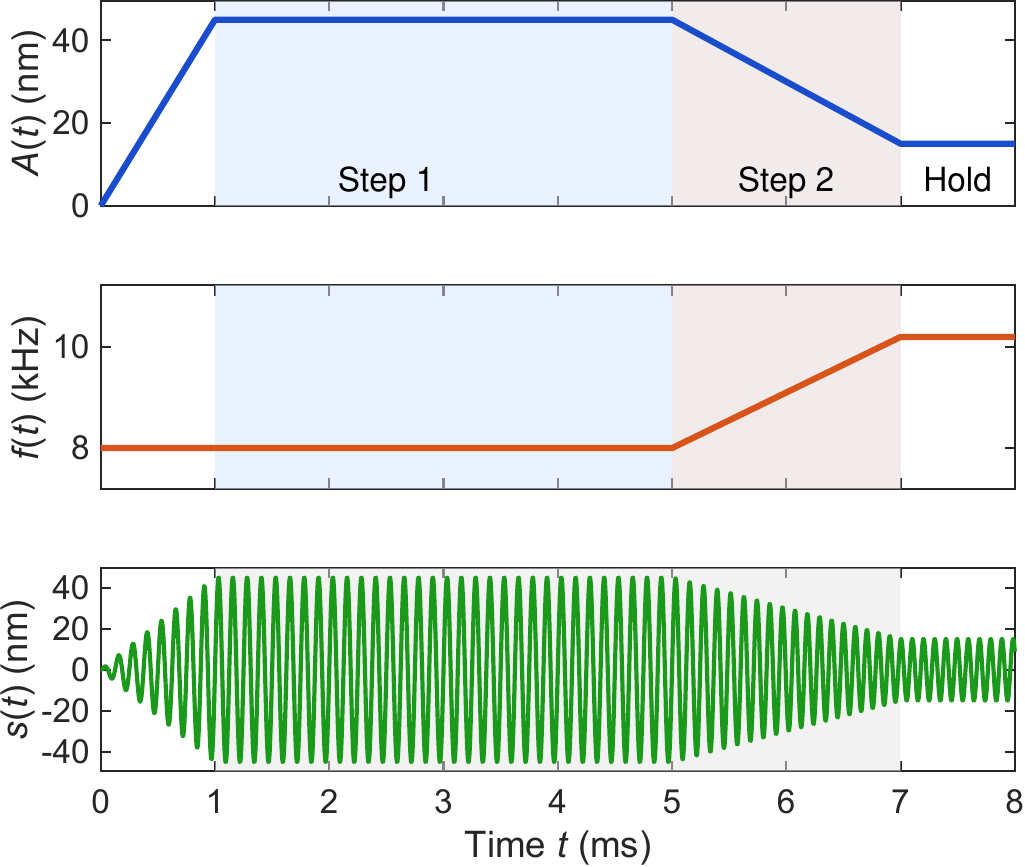}
\caption{Timing sequence of lattice shaking. \textbf{a}–\textbf{c}, Time dependence of the shaking amplitude $A(t)$, frequency $f(t)$, and displacement $s(t)$. As an example, we show the case with initial shaking frequency $f_i=\SI{8}{kHz}$ and amplitude $A_i=\SI{45}{nm}$.}
\label{fig:TimeSequence}
\end{figure}

\section{Calibration of the effective coupling strength}
\begin{figure}[htbp]
\centering
\includegraphics[width=\textwidth]{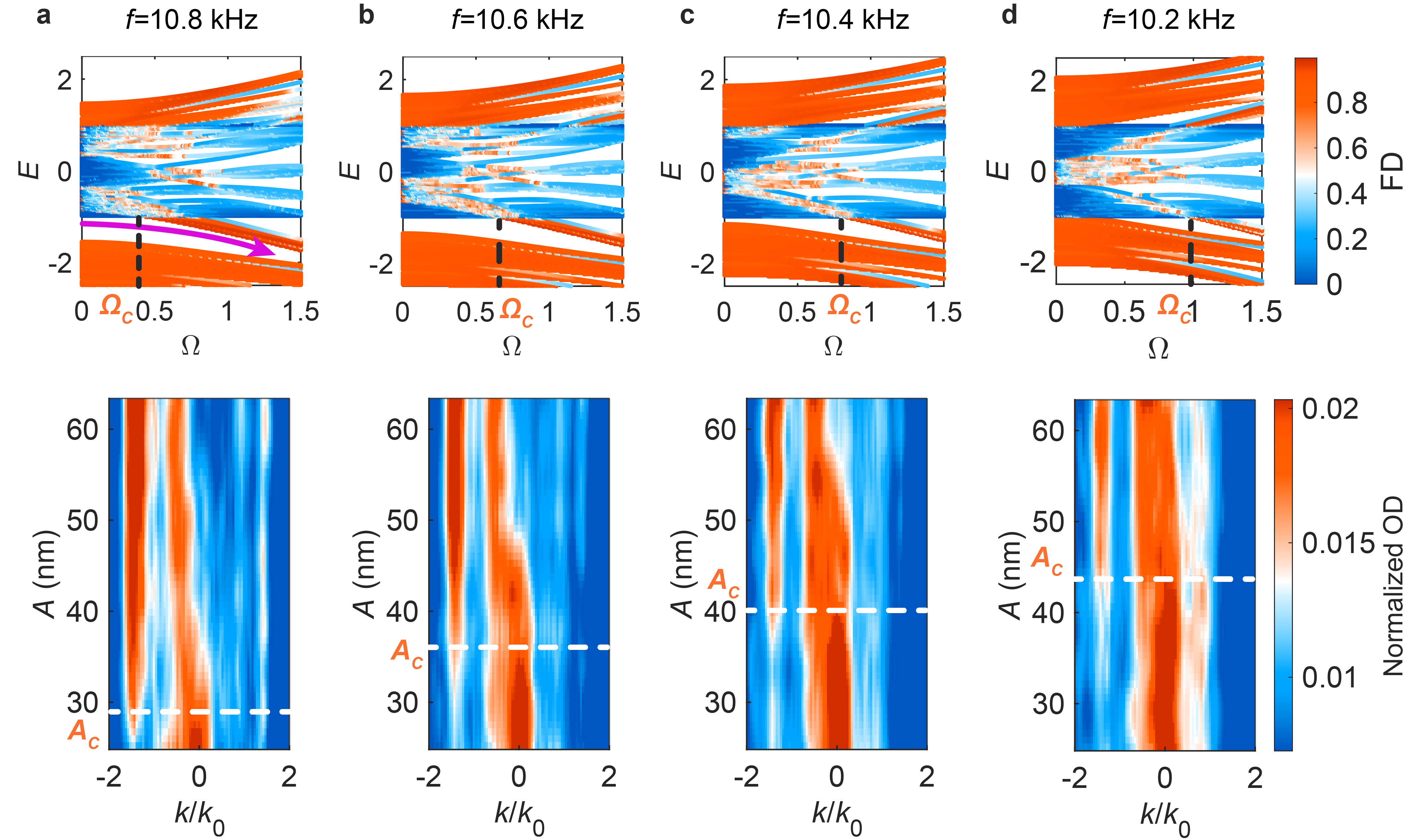}
\caption{Calibration of the effective coupling strength. The shaking frequency is held fixed while the amplitude is linearly ramped from 0 to $A$ within \SI{5}{ms}, followed by time-of-flight (TOF) measurement of the momentum distribution. As the final amplitude $A$ increases, the momentum distribution changes from a localized state to a two-peak distribution. This process corresponds to a Landau–Zener transition in the spectrum, as indicated by the arrow in (\textbf{a}). When the coupling strength exceeds the critical value $\Omega_c$, the state evolves into an extended state with four momentum peaks in the first two Brillouin zones. In the experiment, only two peaks are observed due to Floquet micromotion. Measurements performed at different shaking frequencies yield the critical amplitudes $A_c$, which quantitatively map the experimental shaking amplitude to the effective coupling strength $\Omega$ of the tight-binding model.}
\label{fig:Supple_Omega}
\end{figure}
% In the experiment, we prepare the critical states in the spectrum using a two-step Landau–Zener transition. A fully adiabatic ramp only allows access to the ground state; therefore, nonadiabatic evolution is required to populate excited states.
We experimentally find that at sufficiently large coupling strengths the atoms can be prepared in an excited extended state characterized by four momentum peaks, which we denote as the $E_2$ state. This state corresponds to an excited eigenstate of the effective Hamiltonian, as indicated in Fig.~\ref{fig:Supple_SchematicOfPreparation}a. The $E_2$ state can be prepared either by ramping the shaking amplitude or shaking frequency. Taking ramping the shaking amplitude as an example, in practice we linearly ramp the shaking amplitude from 0 to $A$ within \SI{5}{ms} with shaking frequency fixed, and subsequently perform TOF measurements to obtain the momentum distribution, as indicated in the purple arrow in Fig.~\ref{fig:Supple_Omega}a. The measured momentum distributions as a function of the final amplitude $A$ are shown in the lower row of Fig.~\ref{fig:Supple_Omega}. As $A$ increases, the momentum distribution evolves from the initially localized state to a two-peak distribution. The observed asymmetry, with peaks appearing only on the left-hand side, arises from Floquet micromotion\cite{Eckardt2017Colloquium,Bukov2015high-frequency,Goldman2014PRX,Songbo2022,Sun_2023_micromotion}: the peak positions alternate with the driving phase, and we choose the evolution time such that atoms predominantly condense on the left side. This behavior is consistent with the spectrum obtained from the effective Hamiltonian calculation (upper row of Fig.~\ref{fig:Supple_Omega}), where the initial atoms evolve from the $s$-band ground state to an extended state with increasing coupling strength. Beyond a critical coupling $\Omega_c$, the atoms transfer into this excited extended branch. Moreover, $\Omega_c$ depends on the shaking frequency $f$: as $f$ decreases, $\Omega_c$ shifts to higher values, in agreement with experimental observations. Fig.~\ref{fig:Supple_Omega}a-d show the evolution at different shaking frequencies, demonstrating that the critical shaking amplitude $A_c$ at which the transition occurs systematically increases with decreasing frequency. The model parameters are fixed as $V=J_p=1$ and $\alpha=1.411$, and all energies are expressed in units of $J_p$. Thus, we can build a bridge from our experimental parameters to the tight-binding parameters of the effective Hamiltonian.

\section{State Preparation}
\begin{figure}[htbp]
\centering
\includegraphics[width=\textwidth]{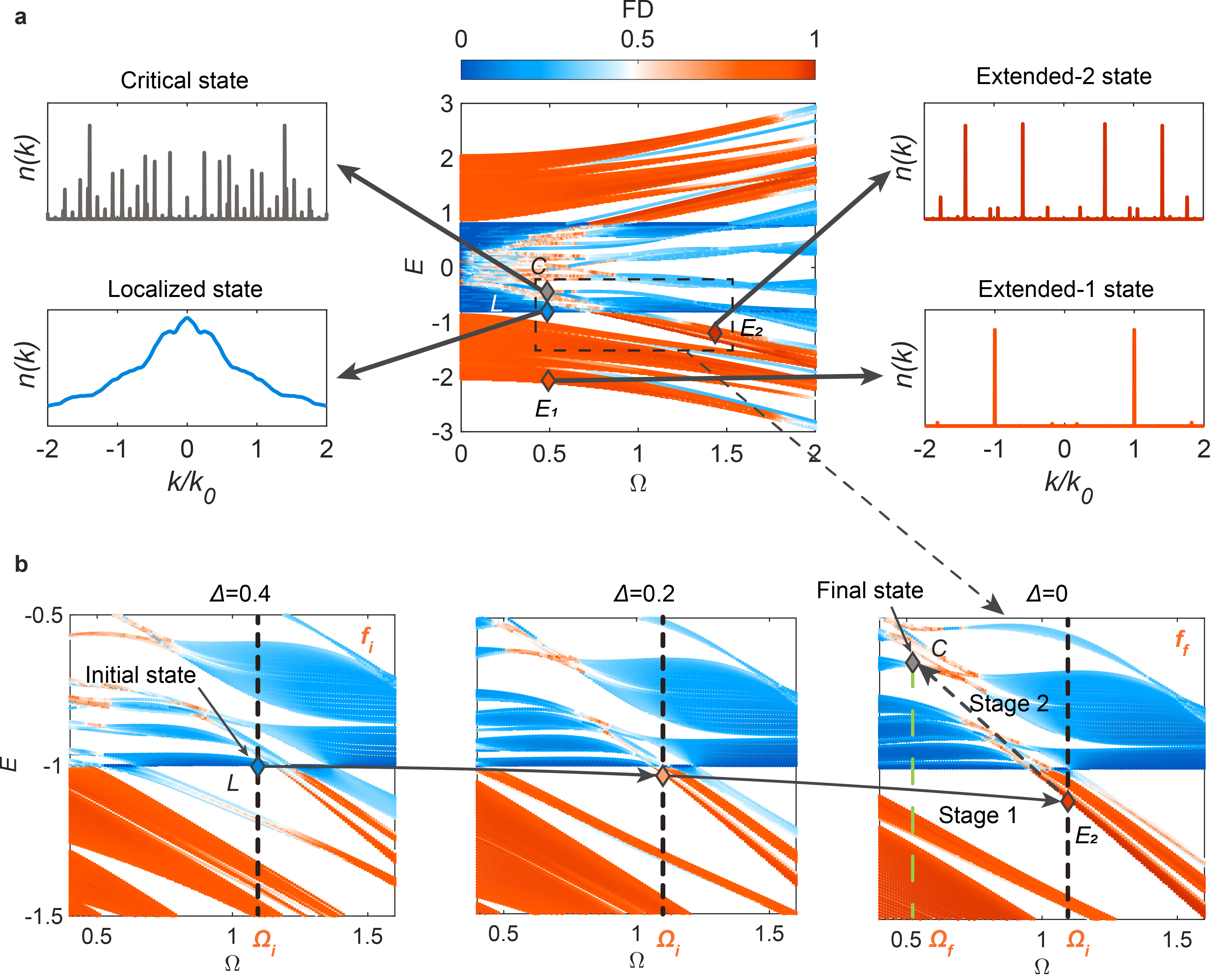}
\caption{Schematics of the two-stage preparation protocol. \textbf{a}, Theoretical momentum distributions of typical eigenstates of Hamiltonian Eq.~(8) (main text). The central panel shows the fractal dimension as a function of the coupling strength, calculated for $\Delta=0$, $V=1$, and $\alpha=1.411$. The four panels on the sides display the momentum distributions corresponding to four representative states in the spectrum. \textbf{b}, Experimental schematic of the two-stage preparation protocol for generating the critical state. From left to right, the panels correspond to fractal-dimension spectra at detunings $\Delta=0.4$, $\Delta=0.2$, and $\Delta=0$, which reflect the stage-1 frequency ramp from $f_i$ to $f_f$. In this process, the initial localized state evolves into the $E_2$ state. In stage 2, a rapid ramp-down of the shaking amplitude drives the state into the critical regime by reducing the coupling strength from $\Omega_i$ to $\Omega_f$, thereby realizing the critical state.}

\label{fig:Supple_SchematicOfPreparation}
\end{figure}

This section provides additional details on the state preparation procedure. The general preparation scheme and the underlying physical picture have been described in the main text. The shaking pulse sequence used in the two-stage protocol is presented in the Methods section and illustrated in Fig.~\ref{fig:TimeSequence}. Here we focus on how different choices of the initial shaking parameters $f_i$ and $A_i$ lead to distinct final states, and provide representative TOF images characterizing each state.

Among the eigenstates of Hamiltonian Eq.~(8) (main text), the fractal dimensions and momentum distributions of four representative states are shown in Fig.~\ref{fig:Supple_SchematicOfPreparation}a. By dynamically ramping the Hamiltonian parameters, we are able to prepare all of these states. The extended ground state ($E_1$) is prepared by sweeping the shaking frequency from an initial $f_i$ to the final frequency. When the coupling is sufficiently strong, the atoms are transferred adiabatically into the bottom of the $p$-band, corresponding to the $E_1$ state.

In contrast, preparation of the critical state requires a two-step process mediated by the excited extended state ($E_2$). As established in the calibration measurements, the $E_2$ state can be prepared either by ramping the shaking amplitude or by ramping the frequency. In the latter case, shown in Fig.~\ref{fig:Supple_SchematicOfPreparation}b, increasing the frequency from $f_i$ to $f_f$ shifts the $E_2$ branch downward in the spectrum, enabling adiabatic transfer of atoms into the $E_2$ state. The critical state lies in the overlap region between this extended branch and the localized $s$-band states. A rapid reduction of the shaking amplitude therefore drives the system nonadiabatically into this overlap regime, realizing the critical state via a Landau--Zener transition, as illustrated in Fig.~\ref{fig:Supple_SchematicOfPreparation}b.

To establish a unified preparation scheme for all three types of states, we employ a two-stage protocol: in stage 1 the shaking frequency is ramped, and in stage 2 the shaking amplitude is rapidly reduced. Specifically, adiabatically ramping the frequency in stage 1 brings the $s$- and $p$-bands into resonance, transferring atoms from the localized $s$-band ground state into the lowest-energy $p$-band state ($E_1$). The outcome of stage 1 depends on the choice of initial frequency $f_i$: sweeping through the $p$-band minimum results in the $E_1$ state, while bypassing the minimum leads to the $E_2$ state. In stage 2, the rapid ramp-down of the amplitude leaves the $E_1$ state unchanged, but drives the $E_2$ state into the critical state.

\begin{figure}[htbp]
\centering
\includegraphics[width=\textwidth]{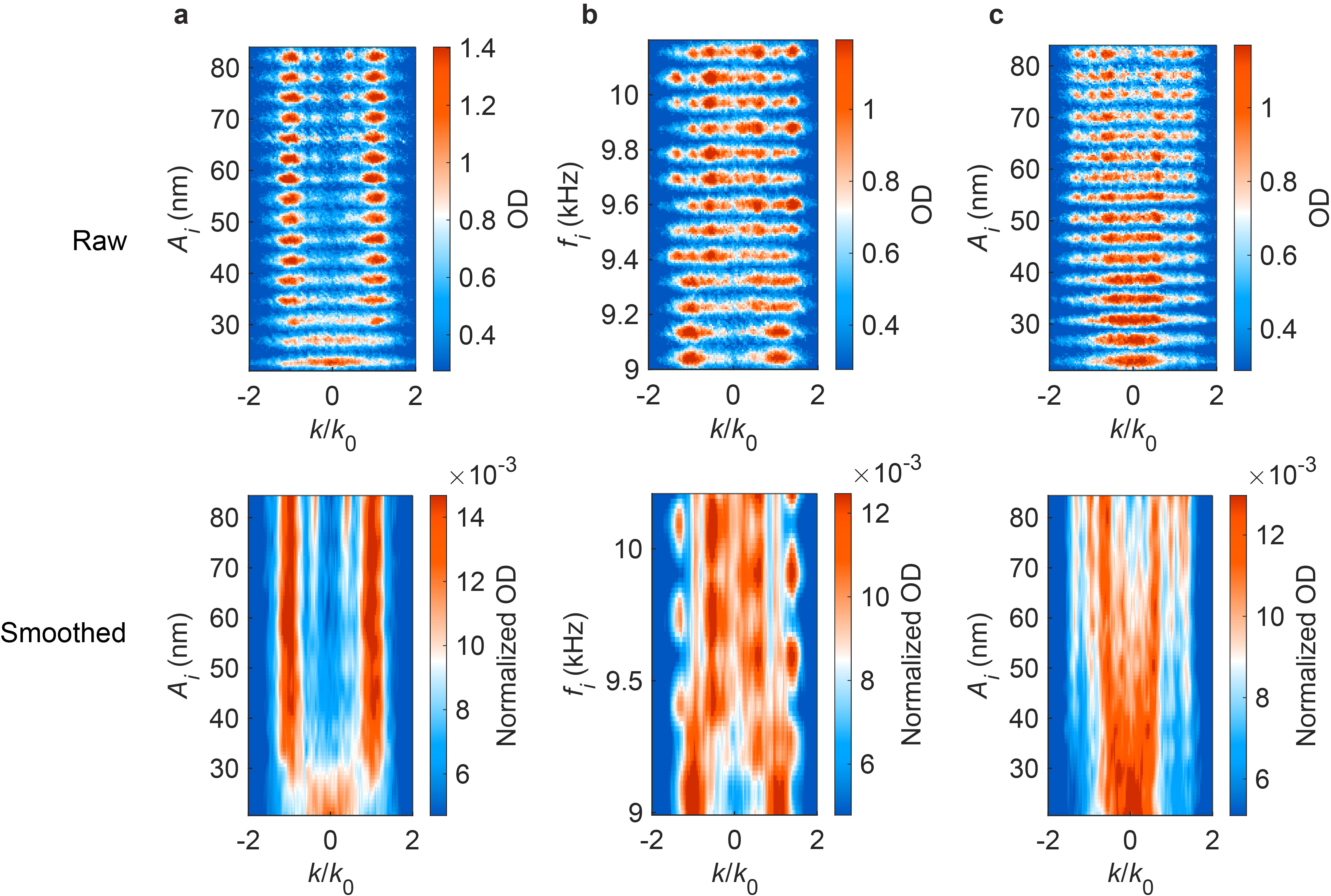}
\caption{Different final states accessed via the two-stage protocol. \textbf{a}, For $f_i = \SI{9}{kHz}$, increasing the initial shaking amplitude $A_i$ drives the final state from localized to extended. \textbf{b}, With fixed $A_i = \SI{63}{nm}$, tuning $f_i$ drives the final state from extended to critical. \textbf{c}, For $f_i = \SI{10}{kHz}$, varying $A_i$ drives the system from localized to critical. Upper rows show raw TOF images; lower rows show the corresponding smoothed profiles used throughout the main text figures.}
\label{fig:Supple_FinalState}
\end{figure}

By tuning $f_i$ and $A_i$, we obtain distinct final states. For small $A_i$, atoms remain in the localized state. For larger $A_i$, the final state evolves into either an extended or a critical state depending on $f_i$. Specifically, with $f_i = \SI{9}{kHz}$, increasing $A_i$ drives a transition from the localized to the extended state [Fig.~\ref{fig:Supple_FinalState}a], whereas with $f_i = \SI{10}{kHz}$ the same variation of $A_i$ drives the system from localized to critical [Fig.~\ref{fig:Supple_FinalState}c]. Fixing $A_i = \SI{63}{nm}$ and varying $f_i$ instead takes the system from extended to critical [Fig.~\ref{fig:Supple_FinalState}b]. These three cases correspond to the three distinct preparation paths illustrated in Fig.~3 of the main text.

Figure~\ref{fig:Supple_FinalState} presents both the raw TOF images (upper row) and the smoothed profiles (lower row); the latter are used in all other figures of the manuscript. The raw OD images are obtained directly from absorption imaging. For each data point, we integrate the OD along the $y$ direction to obtain the one-dimensional momentum distribution $n(k_x)$. This distribution is then normalized within the first two Brillouin zones as
\begin{equation}
n_\mathrm{norm}(k_x) = \frac{n(k_x)}{\int \mathrm{d}k_x\, n(k_x)}.
\end{equation}
The normalized distributions $n_\mathrm{norm}(k_x)$ from different parameter points are subsequently concatenated and interpolated smoothly between parameter values, yielding the smoothed profiles shown in the lower rows.

\section{Numerical Simulation}
\begin{figure}[htbp]
  \centering
    \includegraphics[width=\textwidth]{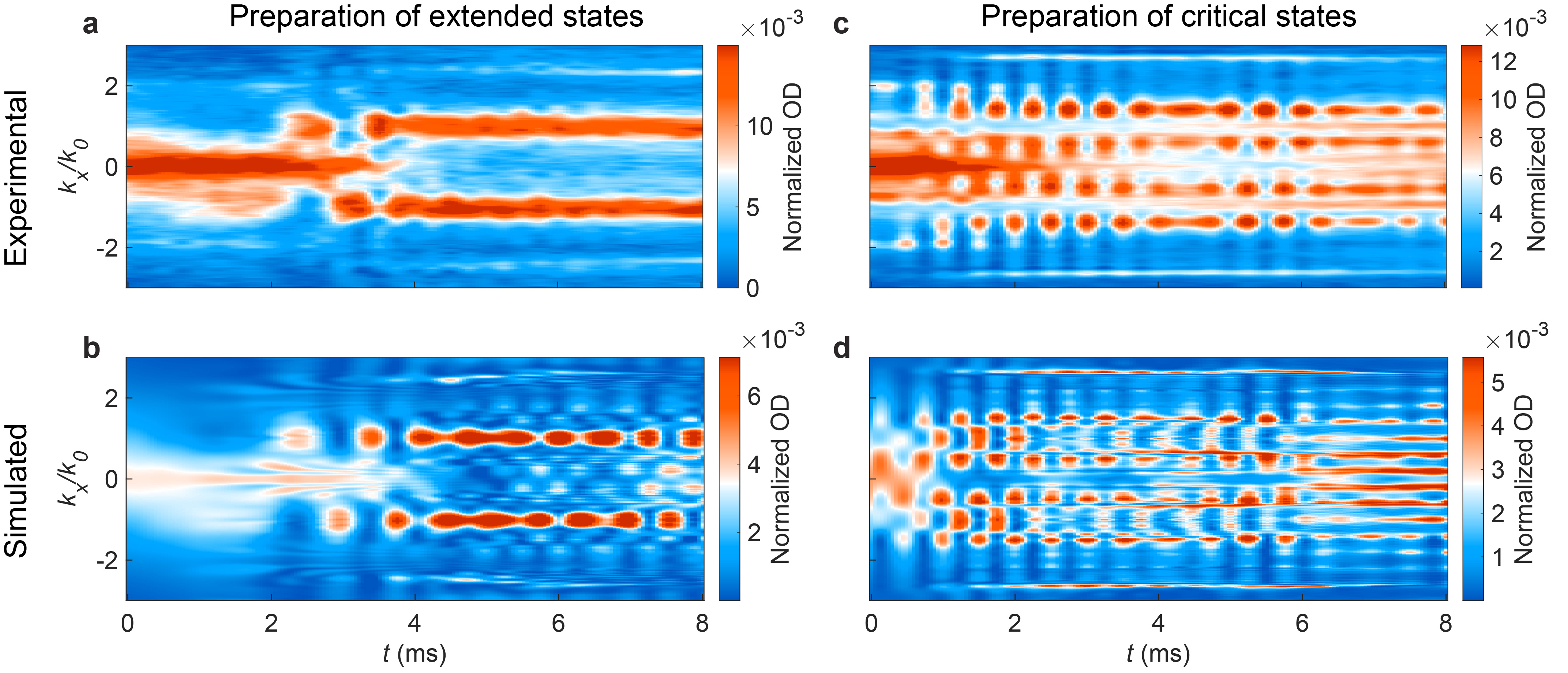}
  \caption{Experimental and simulated time slices of preparation for different states. \textbf{a,b} show the time evolution of the momentum distribution of preparation of the extended state. From \(0\) to \SI{5}{\milli\second}, the system is prepared from the localized state to the extended ground state $E_1$, then from \(5\) to \SI{7}{\milli\second}, the coupling strength is decreased while it remains in the extended ground state, which then also stays stable with a constant Hamiltonian for the last \SI{1}{ms} holding. \textbf{c,d} show the time evolution of the momentum distribution of the critical state. From \(0\) to \SI{5}{\milli\second}, the system is prepared from a localized state to the $E_2$ state, then from \(5\) to \SI{7}{\milli\second} it is prepared into the critical state, which remains stable for the last \SI{1}{ms} with a constant Hamiltonian. Here, \textbf{a},\textbf{c} show the experimental results, while \textbf{b},\textbf{d} present the simulated results, which are in good agreement. The measurements are taken with a time step of \SI{0.25}{ms} for both experimental and simulated results.}
  \label{fig:both}
\end{figure}
We present a theoretical framework to simulate the preparation of distinct quantum states by numerically solving the spatially discretized Schrödinger equation for atoms in a one-dimensional shaken optical lattice. To simplify the system, we model it in one dimension.
\begin{equation}
i\hbar\,\frac{\partial \psi(\mathbf{x},t)}{\partial t} =
e^{-i\gamma}\!\left[-\,\frac{\hbar^2}{2m}\nabla^2 + V_{\text{tot}}(\mathbf{x},t)\right]\psi(\mathbf{x},t),
\end{equation}
where \(\psi(\mathbf{x},t)\) is the condensate wavefunction, \(\hbar\) the reduced Planck constant, \(m\) the atomic mass, and \(\nabla^2\) the Laplacian. Dissipation is phenomenologically incorporated via the factor \(e^{-i\gamma}\)\cite{Anderson2017PRL}. The total potential \(V_{\text{tot}}(\mathbf{x},t)=V_{\mathrm{DT}}+V_{\mathrm{p}}+V_{\mathrm{q}}\) comprises: (i) a harmonic confinement \(V_{\mathrm{DT}}=\tfrac{1}{2}\omega^2 x^2\); (ii) a primary lattice \(V_{\mathrm{p}}(x)=V_{\mathrm{p}}\sin^2(k_{\mathrm{p}}x)\) of wavelength \(\lambda_{\mathrm{p}}=\SI{1064}{\nano\meter}\); and (iii) a quasiperiodic lattice \(V_{\mathrm{q}}(x)=V_{\mathrm{q}}\sin^2(k_{\mathrm{q}}x)\) of wavelength \(\lambda_{\mathrm{q}}=\SI{755}{\nano\meter}\). For completeness, we set \(k_{\alpha}=2\pi/\lambda_{\alpha}\) and \(E_{\alpha}=\hbar^2 k_{\alpha}^2/(2m)\) for \(\alpha\in\{\mathrm{p},\mathrm{q}\}\).

Although the Gross–Pitaevskii equation is commonly employed to model the dynamics of 
interacting Bose–Einstein condensates, we instead adopt the single-particle 
Schrödinger equation in the present study. 
This simplification is justified by the weak transverse confinement of the optical 
potential, which gives rise to a broad radial distribution and therefore a 
significantly reduced three-dimensional atomic density, rendering the mean-field 
interaction energy negligible compared with the total system energy. 
For our parameters, the atomic cloud prior to the shaking process exhibits a 
radial size of approximately \SI{25}{\micro\metre} and spans roughly 30–40 lattice sites 
along the $x$ direction, leading to a dilute spatial density that suppresses 
interaction-induced nonlinearities. 
The mean-field interaction energy can be estimated as 
$E_{\mathrm{int}} = U N_{\mathrm{ave}}$, where $N_{\mathrm{ave}}$ is the average atom number 
per lattice site and 
$U = \frac{4\pi \hbar^{2} a_s}{m}\int |\psi(\mathbf r)|^{4}\,d^{3}\mathbf r$. 
By approximating $\psi(\mathbf r)$ with a noninteracting Gaussian wavefunction of the 
measured transverse width, we find that the mean-field contribution amounts to less 
than \SI{1}{\percent} of the total energy. 
Consequently, the interaction effect remains perturbative throughout the dynamics, 
and the essential physics can be captured within a non-interacting 
single-particle framework.

The initial state at \(t=0\) is obtained by solving the stationary Schrödinger equation with parameters matched to the experiment: \(\omega=\SI{20}{\hertz}\), \(V_{\mathrm{p}}=8.7E_{\mathrm{p}}\), and \(V_{\mathrm{q}}=1.6E_{\mathrm{q}}\). We simulate the subsequent momentum evolution using the experimental driving sequence, represented by a time-dependent displacement of the quasiperiodic lattice. During the first \SI{1}{\milli\second}, the extended and critical protocols are initiated at driving frequencies of \SI{8}{\kilo\hertz} and \SI{10}{\kilo\hertz}, respectively, while the shaking amplitude is ramped from zero to \(\SI{63}{\nano\meter}\). From \SI{1}{\milli\second} to \SI{5}{\milli\second}, the amplitude is kept constant as the frequency is linearly increased to \SI{10.8}{\kilo\hertz}, steering the system along the two protocols into the intermediate \(E_1\) and \(E_2\) manifolds. Between \SI{5}{\milli\second} and \SI{7}{\milli\second}, the frequency remains fixed while the amplitude is reduced to \(\SI{21}{\nano\meter}\). In the final \SI{1}{\milli\second}, both the drive frequency and amplitude are held constant, rendering the Hamiltonian time-independent. The simulated momentum distributions for these preparation sequences, presented in Fig.~\ref{fig:both}, show good overall agreement with the experimental data, indicating that the main features of the observed dynamics are reproduced by the simulations.

For clarity of presentation, we select a representative set of parameters that most transparently reveals the underlying physical behavior.  
Although the values of \(V_p\) and the shaking frequency used in the numerical simulations differ slightly from those used in the experiment, both experimental observations and theoretical analyses demonstrate that the global phase diagram and the characteristic time-slice dynamics remain essentially unchanged over a broad parameter range. This robustness ensures that the minor parameter mismatch does not affect the physical interpretation of the results.  
Such discrepancies between experiment and simulation arise from slow optical-path drifts, slight miscalibrations of the lattice depth, limited accuracy of the phenomenological dissipation term, or other uncontrolled experimental factors. These effects, though unavoidable in realistic conditions, are secondary in nature and do not alter the overall agreement between theory and experiment.

\section{Extraction and Analysis of Characteristic Length Scales}
\label{sec:appendix_length_scales}

To quantitatively differentiate the localized, extended, and critical states prepared via our two-stage protocol, we employ a dual-space perspective utilizing two complementary characteristic length scales: the spatial coherence length $L_{\mathrm{Coh}}$ and the momentum-space correlation length $k_{\mathrm{Cor}}$. Because a critical state exhibits a multifractal wave function that evades standard localization in both real and reciprocal domains, relying on a single observable is insufficient. By simultaneously mapping the behavior of both $L_{\mathrm{Coh}}$ and $k_{\mathrm{Cor}}$, we establish a robust framework to resolve the tripartite phase diagram.

The first observable, $L_{\mathrm{Coh}}$, quantifies the characteristic spatial correlations and the degree of localization within the atomic ensemble. Invoking the Wiener-Khinchin theorem, the momentum distribution $n(k)$ is fundamentally linked to the first-order spatial correlation function $G\left(x^{\prime}, x+x^{\prime}\right)=\langle\hat{\Psi}^{\dagger}\left(x^{\prime}\right) \hat{\Psi}\left(x+x^{\prime}\right)\rangle$. Specifically, the power spectrum of the field (the momentum density) is the Fourier transform of its autocorrelation, expressed through the relation \cite{Deissler2010NP}:
\begin{equation}
	n(k) \propto \mathfrak{F}^{-1} \left[ \int G\left(x^{\prime}, x+x^{\prime}\right) \mathrm{d} x^{\prime} \right],
\end{equation}
where $\mathfrak{F}$ denotes the Fourier transform operator. In our analysis, we utilize the inverse relation to extract the spatial information from the experimental observables. We define $S(x) = \mathfrak{F}[n(k)]$, which represents the ensemble-averaged spatial correlation. The spatial coherence length is then defined as the normalized first moment of the absolute correlation function:
\begin{equation}
	L_{\mathrm{Coh}}=\frac{\int_0^{x_\mathrm{max}} \mathrm{d}x \, |S(x) x|}{\int_0^{x_\mathrm{max}} \mathrm{d}x \, |S(x)|}.
\end{equation}
Physically, $L_{\mathrm{Coh}}$ reflects the competition between the intrinsic localization length of the wave functions and the many-body coherence of the system. In the experiment, the extracted absolute values of $L_{\mathrm{Coh}}$ are inherently smaller than those of an ideal pure eigenstate. This reduction is primarily driven by thermalization and decoherence effects, as the macroscopic atomic gas occupies a statistical manifold of multiple states. Nevertheless, the relative distinction between the localized and extended phases is dominantly governed by the intrinsic spatial extent of the underlying constituent wave functions. Therefore, $L_{\mathrm{Coh}}$ effectively captures the localization-to-delocalization transitions, enabling us to differentiate the quantum states despite the many-body dephasing.

In practice, the extraction of $L_{\mathrm{Coh}}$ involves a specific data processing pipeline to isolate the lattice physics from external confinement effects. First, the raw two-dimensional Time-of-Flight (TOF) images are integrated along the transverse direction to yield the 1D momentum density $n(k)$. After computing $S(x)$, we evaluate the integral over a truncated domain bounded by $x_{\mathrm{max}} = 5d$ (where $d = \SI{532}{nm}$ is the primary lattice constant). This choice of $x_{\mathrm{max}}$ is crucial: it is significantly smaller than the typical characteristic length of the global harmonic confinement ($\approx \SI{2.4}{\micro\metre}$), ensuring that the captured signal is not dominated by the external trap geometry. To further align with the discrete symmetry of the lattice, we retain only the correlation amplitudes at integer lattice sites $x_j = jd$ ($j \in \{0, 1, \dots, 5\}$), resulting in a discrete summation for the final calculation of $L_{\mathrm{Coh}}$.

Complementary to the real-space analysis, the density-density correlation length $k_{\mathrm{Cor}}$ is introduced to quantify the characteristic spread and fragmentation of the quantum states in reciprocal space. We first calculate the momentum-space autocorrelation function $C(k)$, defined as:
\begin{equation}
	C(k) = \int_{-2k_0}^{2k_0} n(k'+k) n(k') \, \mathrm{d}k',
\end{equation}
where $k_0$ represents the characteristic momentum scale of the lattice. The correlation length is subsequently extracted through the normalized first moment of $C(k)$:
\begin{equation}
	k_{\mathrm{Cor}} = \frac{\int_0^{k_0} |C(k)k| \, \mathrm{d}k}{\int_0^{k_0} |C(k)| \, \mathrm{d}k}.
\end{equation}
This observable provides a sensitive probe for the underlying momentum structure. An extended state, characterized by well-defined and sharp Bragg peaks, yields a highly localized autocorrelation profile $C(k)$, thereby minimizing $k_{\mathrm{Cor}}$. In contrast, both localized and critical states lack long-range spatial periodicity, resulting in broad, highly fragmented momentum profiles. This reciprocal-space spreading drives a substantially larger $k_{\mathrm{Cor}}$, establishing it as an essential metric for distinguishing fully extended states from the rest of the phase diagram.

By synthesizing the behaviors of $L_{\mathrm{Coh}}$ and $k_{\mathrm{Cor}}$, we identify Region I as the localized phase (suppressed $L_{\mathrm{Coh}}$, large $k_{\mathrm{Cor}}$), Region II as the extended phase (large $L_{\mathrm{Coh}}$, minimized $k_{\mathrm{Cor}}$), and Region III as the critical phase, which features a unique "dual-delocalization" signature (simultaneously large $L_{\mathrm{Coh}}$ and $k_{\mathrm{Cor}}$). To rigorously validate these experimental extractions, we benchmark our results against numerical simulations as detailed in the preceding section. By adopting parameters that directly match the experimental conditions, we calculate the theoretical final momentum distributions. Crucially, both the theoretical and experimental momentum density profiles are subjected to the exact same data processing pipeline—including the identical spatial truncation and discrete integration methodology—to extract the characteristic length scales. This systematic comparison demonstrates good agreement, confirming that the dual-space characteristic functions reliably map the tripartite phase diagram.

\section{Micromotion of the States}
\begin{figure}[htbp]
  \centering
    \includegraphics[width=\textwidth]{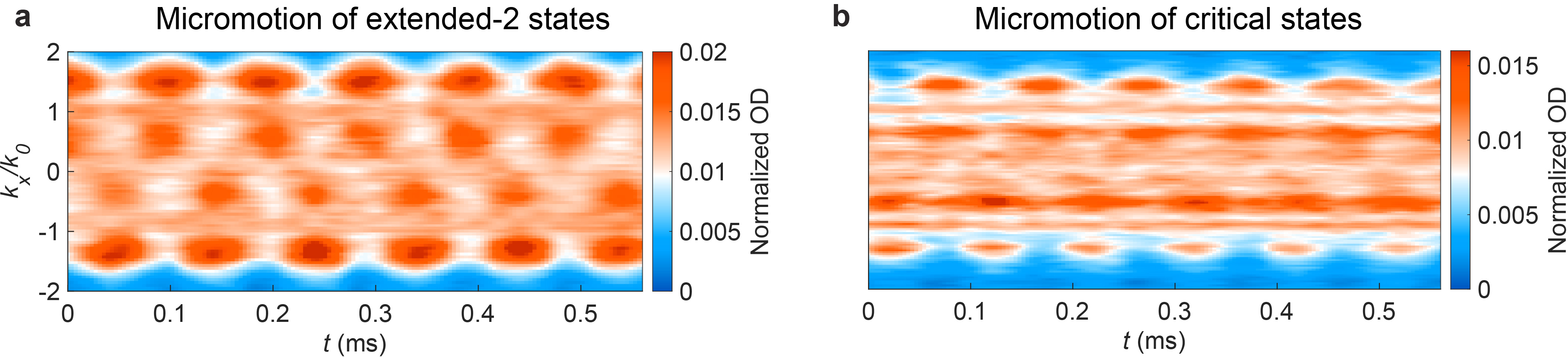}
  \caption{Experimental results for the $E_2$ state (\textbf{a}) and the critical state (\textbf{b}). After preparing the atoms in these states, we hold the Hamiltonian and measure the momentum distribution as a function of time. In both cases, the momentum peaks alternate between the left and right sides, oscillating at the final driving frequency $f_f=\SI{10.2}{kHz}$. The measurements are taken with a time step of \SI{0.02}{ms}, and each data point represents the average of four images. For clarity, the images are further processed by smoothing.}
  \label{fig:micromotion}
\end{figure}
In a Floquet-driven system, the atomic states are characterized not only by their long-term stroboscopic evolution but also by a fast periodic evolution within each driving cycle, known as micromotion~\cite{Eckardt2017Colloquium,Bukov2015high-frequency,Goldman2014PRX,Sun_2023_micromotion}. While the effective Hamiltonian $\hat{H}_{\mathrm{eff}}$ captures the time-independent physics when the system is sampled at integer multiples of the driving period $T = 1/f_f$, the true time-dependent state $|\psi(t)\rangle$ continues to evolve according to the full Hamiltonian $\hat{H}(t)$. This micromotion reflects the periodic population exchange and phase evolution between the coupled $s$ and $p$ bands, manifesting as high-frequency oscillations of observable quantities at the modulation frequency.

To experimentally resolve these fast dynamics, we perform TOF measurements with high temporal resolution after the adiabatic preparation sequence. Once the system reaches the target state, we hold the Hamiltonian parameters constant and record the momentum distribution at various time delays within the driving cycle. As illustrated in Fig.~\ref{fig:micromotion}, the micromotion is clearly manifested as a periodic redistribution of the momentum density. Specifically, the momentum peaks are observed to alternate between the left and right sides of the Brillouin zone, oscillating precisely at the final driving frequency $f_f = \SI{10.2}{kHz}$. This "sloshing" motion in momentum space is a direct consequence of the time-dependent interference between the $s$ and $p$ orbital components, which possess different spatial parities.

The experimental results for both the $E_2$ state (an excited extended state) and the critical state are shown in Fig.~\ref{fig:micromotion}a and Fig.~\ref{fig:micromotion}b, respectively. Both states exhibit pronounced micromotion, which serves as a sensitive probe of their orbital composition. The observation of these oscillations confirms that these states are not merely stationary states of the primary lattice, but are instead dynamic dressed states emerging from the strong, resonant $s$–$p$ coupling induced by the shaken quasiperiodic potential. Each data point in these plots represents an average of four independent images taken with a fine time step of \SI{0.02}{ms} (approximately $1/5$ of a driving period), ensuring that the fast oscillations are accurately captured without aliasing. The high contrast of the micromotion in the critical regime further underscores the significant role of interband hybridization in the formation of the tripartite quantum phase.

\section{Expansion Dynamics}

We experimentally observe that the density distribution after a finite expansion time exhibits a bimodal profile: a fraction of atoms expands rapidly while the remainder stays localized. A likely origin is incomplete adiabatic transfer during the state-preparation sequence, leaving a portion of atoms in the initial $s$-band localized state rather than being excited into the $p$-band coupled regime. The relative fraction of atoms successfully prepared into the target state depends on the amplitude of the first shaking stage $A_i$, as shown in Fig.~4d of the main text.

The results obtained by varying different preparation parameters are summarized in Fig.~\ref{fig:Exp}. The data used for Figs.~4e,f in the main text are extracted from these measurements. As shown in Fig.~\ref{fig:Exp}a, increasing $A_i$ at fixed $f_i$ drives the final state from localized to critical. In the localized regime, atoms show negligible expansion; in the critical regime, the expansion radius grows as $t^{0.5}$. In Fig.~\ref{fig:Exp}b, we fix a relatively large $A_i$ and vary $f_i$, driving the system from the extended to the critical regime. The dynamical exponent extracted from the time dependence changes continuously from $1$ to $0.5$, consistent with theoretical predictions for the three distinct transport phases~\cite{Wang2022PRB,Saha2019PRB}.

\begin{figure}[htbp]
\centering
\includegraphics[width=0.8\textwidth]{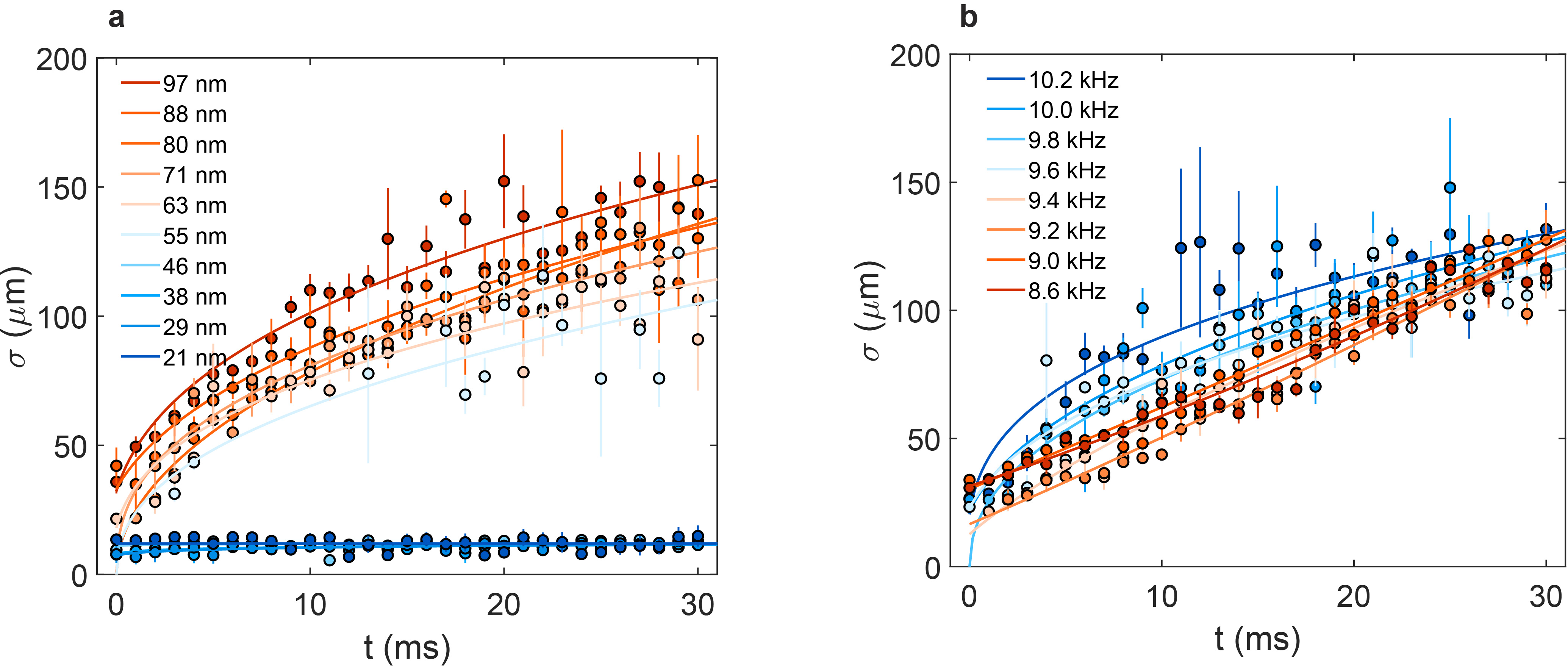}
\caption{Expansion dynamics under different preparation parameters. \textbf{a}, Measured radius as a function of expansion time for varying $A_i$ at fixed $f_i = \SI{10}{kHz}$. For small $A_i$ the atoms remain localized with negligible expansion, whereas larger $A_i$ drives the system into the critical regime with diffusive spreading. \textbf{b}, Measured radius for varying $f_i$ at fixed $A_i = \SI{75}{nm}$. As $f_i$ increases, the expansion crosses over from approximately linear to square-root scaling in time. Dots are experimental data; solid lines are power-law fits.}
\label{fig:Exp}
\end{figure}

For the expansion measurements we use a characteristic timescale of \SI{30}{ms}, within which clear broadening of the atomic cloud is already visible. The relatively large $p$-band hopping amplitude ($J_p \sim \SI{0.5}{kHz}$) ensures that the expansion is significant on this timescale. At longer times, drive-induced heating and decoherence thermalize the system, causing the diffusion dynamics of all states to converge, as shown in Fig.~\ref{fig:x2moment}b.

We also monitor the total atom number throughout the evolution, as shown in Fig.~\ref{fig:x2moment}a. All three states exhibit lifetimes exceeding \SI{200}{ms}, confirming that atom loss is negligible within the observation window.

To characterize the diffusion quantitatively, we record the in-situ normalized density distribution $n(x,t)$ after holding the Hamiltonian for time $t$. Density profiles are normalized such that $\int n(x,t)\,\mathrm{d}x = 1$, and pixels with optical density below zero are discarded to suppress background noise. From the double-Gaussian fits described in the Methods, we compute the second spatial moment
\begin{equation}
\langle x^2 \rangle = \int_{x_\mathrm{min}}^{x_\mathrm{max}} n_\mathrm{fit}(x,t)\,(x - \bar{x})^2\,\mathrm{d}x,
\end{equation}
where $\bar{x}$ is the mean position and the integration window is $\pm\SI{200}{\mu m}$ around the cloud center, chosen to include the full expanding distribution.

During the first \SI{30}{ms}, the extended, critical, and localized states display clearly distinct diffusion behaviors. Beyond \SI{30}{ms} the dynamics of all three states gradually converge, primarily due to drive-induced heating, background collisions, and partial decay of atoms back into the $s$-band. Accordingly, power-law fits to the expansion data are restricted to the first \SI{30}{ms}, where the measured dynamics faithfully reflect the transport properties of the prepared state.

\begin{figure}[htbp]
\centering
\includegraphics[width=0.8\textwidth]{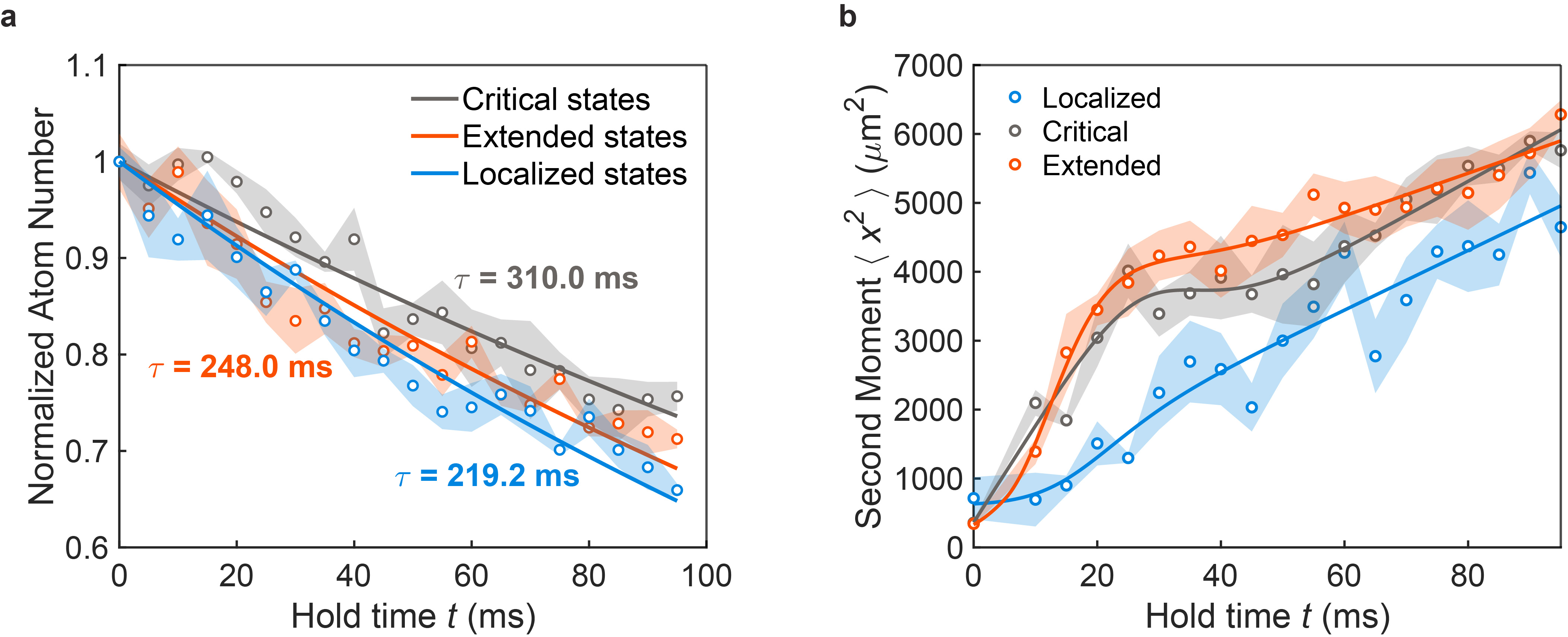}
\caption{Long-time atom number and diffusion dynamics for atoms prepared in different final states.
\textbf{a}, Normalized total atom number as a function of time for the extended, critical, and localized states, prepared with parameters $(A_i,\,f_i) = (\SI{63}{nm},\,\SI{9}{kHz})$, $(\SI{63}{nm},\,\SI{10}{kHz})$, and $(\SI{21}{nm},\,\SI{10}{kHz})$, respectively. All three states retain more than $1/e$ of their initial atom number beyond \SI{100}{ms}. Solid lines are exponential decay fits.
\textbf{b}, Long-time evolution of the second spatial moment $\langle x^2\rangle$, calculated from normalized in-situ density profiles within a symmetric window of $\pm\SI{200}{\mu m}$. The three states show distinct diffusion behaviors during the first \SI{30}{ms}, after which the curves converge due to classical thermalization. Solid lines are guides to the eye.}
\label{fig:x2moment}
\end{figure}

% \normalem
% \bibliographystyle{naturemag}
% \nolinenumbers
% \bibliography{ref}

%%%%%%%%%%%%%%%%%%%%%%%%%%%%%%%%%%%%%

\end{document}